\documentclass[%
 reprint,
superscriptaddress,
 amsmath,amssymb,
 aps,
pra,
longbibliography
]{revtex4-2}

\usepackage{graphicx}% Include figure files
\usepackage{dcolumn}% Align table columns on decimal point
\usepackage{bm}% bold math
\usepackage{hyperref}% add hypertext capabilities

\usepackage{cleveref}
\usepackage{{comment}}
\Crefname{equation}{Eq.}{Eqs.}

\usepackage{amsthm}
\usepackage{thmtools}

\newtheorem{lemma}{Lemma}

\usepackage{color}

\usepackage[normalem]{ulem}
\DeclareRobustCommand\mpwS[1]{{\let\helpcmd\@firstofone\parhelp#1\par\relax\relax} }
\long\def\parhelp#1\par#2\relax{%
	\helpcmd{#1}\ifx\relax#2\else\par\parhelp#2\relax\fi%
}

\newcommand{\mpwh}[1]{{\unskip}}%\hspace{-0.17cm}}}

\begin{document}

\preprint{APS/123-QED}

\title{Measuring Clock Precision Without an Ideal Time Reference\\
%Precisely determining a clock's precision without a more precise reference clock\\
%Precisely determining a clock's precision without precise reference clocks
}

\author{Pablo Á. Domínguez}
 \affiliation{University Grenoble Alpes, Grenoble, France}
 \affiliation{Institute for Theoretical Physics, ETH Zürich, Switzerland}

\author{Christopher Chubb}
\affiliation{Independent Researcher}

\author{Mischa P. Woods}
\affiliation{ENS Lyon, Inria, France}
\affiliation{University Grenoble Alpes, Inria, Grenoble, France}


\renewcommand{\thesubsection}{\thesection.\Alph{subsection}}
\renewcommand{\theHequation}{\theHsection.\arabic{equation}}
\makeatletter
\renewcommand{\p@subsection}{}
\makeatother


\date{\today}% It is always \today, today,

% Include institute/department
% Eu: ETHZ, UGA, (Inria?)
% Christopher: ETHZ
% Mischa: Inria, UGA, ENS Lyon

\begin{abstract}
Characterising physical devices typically requires comparison with more precise reference standards. For clocks, however, this approach is problematic because no physically perfect time reference exists; real clocks have finite, model-dependent uncertainty. We address whether multiple imperfect clocks can collectively determine their precision without an absolute reference. This paper provides analytic bounds and protocols for quantifying clock precision using only multiple copies of the same clock, which need not be synchronised and may have manufacturing variations or environmental coupling. Beyond clock characterisation, these methods enable precise measurements of sub-period intervals and relativistic effects such as time dilation and Doppler shifts, providing tools for quantum metrology and fundamental physics.
\end{abstract}

\maketitle

\section{Introduction}
% Measuring quantities (length, electricity...) leads to tolerances and errors
The concept of measurement forms the bedrock of physics, serving as our primary means of extracting quantifiable information from the natural world to build predictive models and theoretical frameworks. When conducting any measurement, determining the associated uncertainty is essential, accounting for both instrumental limitations and methodological constraints that define the measurement's precision. Consider the straightforward example of measuring distance with a ruler: the uncertainty in the result depends on how precisely the ruler was calibrated during manufacturing, the resolution of its smallest markings, and any systematic errors introduced during the measurement process. This fundamental relationship between measurement and uncertainty underscores the foundation of experimental physics, where the validity of scientific claims must be evaluated within the context of their quantifiable precision limits \cite{giovannetti2011advances, polino2020photonic}.

% Same happens with time. Definition of a clock (just citing your paper about the axiomatic principles plus some short intuitive description, nothing too long)
Time, like any other physical quantity, must be accessed through measurement, requiring specialised instruments known as clocks.
In the context of this paper, a clock is defined broadly as any device capable of autonomously transmitting information about the passage of time to the environment in the form of a sequence of discrete events, which we will refer to as ``ticks''. Due to the quantum nature of time, the tick times are stochastic and can be modelled as random variables with a well-defined mean $\mu$ and finite variance $\sigma^2$. We therefore characterise clocks by the statistics of their inter-tick intervals \cite{Woods_2021,Silva2023TickingClocks}. This inclusive definition encompasses not only conventional timepieces but also specialised instruments such as timers, stopwatches and metronomes. More formally, we will consider a clock to be any physical system that satisfies axioms (1) through (4) of \cite{Woods_2021}, establishing a rigorous mathematical foundation for our subsequent analysis of temporal measurement.

% People want more and more accurate clocks, but there is not an ideal reference frame. In the light of this, we have tried to get as close as possible to a perfect time reference (proof of this is the redefinition of a second in terms of more and more precise oscillators) but never achieved (they require infinite energy; not only in practice but in energy as well due to uncertainty principle)
As humans, we rely on time measurement for countless activities, and for many of these, measurement precision becomes critical. Throughout history, humanity has developed increasingly sophisticated timekeeping devices, as reflected, for example, in the evolution of how we define the second, progressing from pendulums to atomic transitions. However, most current methods for operationally determining the precision of timekeeping devices face one of two fundamental limitations: either they require comparison with a reference clock that is more precise than the clock under test, or they require a mechanism to perform phase measurements or locking \cite{natureClocks}. The former creates a circular dependency, since all clocks ultimately seem to have to be validated against other clocks \cite{meier2023fundamental, Dawkins2007FreqCounter}. Moreover, the latter restricts our ability to measure the precision of the timekeeping device to the precision of the phase measurement or locking mechanism. Because the \textit{phase} (e.g. as in the Pegg-Barnett formalism \cite{pegg1989phase}) and number operators in quantum mechanics do not commute, this will always be tied to a non-zero uncertainty \cite{vaccaro2003superpositions}. 

The Allan deviation \cite{Allan1981ModifiedAllan,Allan1987TimeAF} $\sigma_A(\tau)$ addresses part of this challenge by providing a convergent statistical measure that relaxes the requirement of an ideal time reference: in its two-way measurement form, it allows the characterisation of a timekeeping device by comparing it with a copy across successive measurement intervals over a time $\tau$. While this approach allows the precision of a clock to be measured without direct comparison with an ideal reference \cite{Allan1981ModifiedAllan,Howe2000TotalDev}, it still has significant limitations.

First, by definition, the Allan deviation is measured over a given time interval $t\in(0,\tau]$. In the absence of an ideal clock or reference, one can only approximate this parameter with a precision equal to that of the timekeeping device under test. Second, evaluating the Allan deviation requires frequency measurements. Common methods rely either on frequency counters (which in practice need a timing standard to count cycles per unit time) or on phase/frequency locking (whose precision is limited by the phase measurement itself \cite{Howe2000TotalDev, Dawkins2007FreqCounter, Kirchner1999TWSTFT}). Finally, in order to measure the Allan variance with two copies of a clock, both need to be started simultaneously. One may object that producing a trigger that starts two clocks nearly simultaneously could itself require a timing reference more precise than the devices under test. In short, it is not clear whether it is actually possible to measure the Allan deviation with arbitrary accuracy without depending on an external time reference frame, since the inclusion of any external measurement device may implicitly require external references or additional time or phase measurement devices.

An additional issue is that, one may argue that, from a theoretical point of view, manufacturing two exact copies of the same clock is impossible, thereby adding a new source of uncertainty that may not be easy to incorporate into theoretical models of clocks.

From a purely theoretical perspective, an absolute precision measurement defined using solely the statistical moments of the probability distribution of the clock would offer distinct advantages, such as directly characterising the inherent uncertainty both at short and long measurement periods, requiring no external reference, and, probably most importantly, resulting in a figure of merit for the precision of a clock that does not directly depend on any physical experiment or implementation. This would make it a more convenient theoretical tool than the Allan variance currently is \cite{Woods2022QuantumClocks,natureClocks}.

% Different approach: do we actually need an ideal reference frame to measure time? Can't we manage without one? Can we measure time and tolerance aka precision in an \textbf{operational} way without an ideal reference frame? (not lots of papers focusing on this)
Due to the fundamental nature of time, it seems a priori impossible to operationally evaluate any clock in such an absolute manner without reference to another ideal or more precise clock. Nevertheless, proving or disproving this limitation from theoretical foundations remains challenging. This led us to ask: might there be an alternative approach to clock characterisation and time measurement? Is it truly impossible to measure time and quantify timekeeping precision in an operational manner without requiring an ideal reference frame? These questions motivate the investigation of new methods that might circumvent the apparent necessity of an external temporal standard when evaluating the performance of timekeeping devices \cite{Nurgalieva2024Tomography}.

%Talk about R and E[t] generally, not in detail, to motivate the rest of the text. Also motivate R (all the theoretical papers on clocks bound R and talk about R.
An additional complexity arises from the multiple figures of merit used in the literature for characterising the precision of timekeeping devices. One such figure of merit is $R$, which, for a ticking clock whose (stochastic) ticking period is modelled by the random variable $T$, is theoretically defined as the ratio between its squared mean $\mu_T^2$ and its variance $\sigma_T^2$. This measure is invaluable for deriving theoretical results regarding temporal measurement capabilities and serves as the standard metric in quantum information theory and quantum foundations \cite{Woods2022QuantumClocks, Nurgalieva2024Tomography, Yang2019AccuracyEnhancing, Erker2017AutonomousClocks, moreira2025precisionboundsmultiplecurrents, prech2024optimaltimeestimationclock}.

However, the definition of $R$ raises concerns about the feasibility of performing operational measurements with arbitrary precision in the absence of an ideal time reference frame, since both $\mu_T$ and $\sigma_T$ seemingly require such a reference for their accurate determination. Take $\mu_T$, for example. Its definition suggests that we should record the times at which ticks are observed using a perfect time reference, then compute the mean value of these observed times.

To address this challenge, researchers \cite{doubletick} developed an alternative figure of merit, which we denote as $\mathbb{E}[\tau]$. This metric is defined experimentally by using two identical copies of the device under evaluation and conducting the experiment detailed in \Cref{sec:theoframe}. While this figure of merit can be obtained operationally without an external reference, establishing a clear relationship between $\mathbb{E}[\tau]$ and $R$ remains non-trivial, and $\mathbb{E}[\tau]$ lacks the formal utility needed for deriving rigorous theoretical results about timekeeping capabilities.

This paper therefore proposes new methods for clock characterisation that do not rely on external time reference frames. In addition, we establish a precise, bounded relationship between $R$ and $\mathbb{E}[\tau]$, demonstrating the operational significance of $R$, and we show how precise timekeeping measurements can be performed without requiring an ideal temporal reference. Moreover, we relax the requirement of two exact copies of the clock by requiring only that the two clocks have equal precision, allowing for potentially different oscillation frequencies due to manufacturing defects or relativistic effects \cite{naturePreciseClock,naturePreciseClock2}.

The framework developed here provides the conceptual starting point for our companion work~\cite{Dominguez2026Characterising}. A necessary first step is to establish that the theoretically defined precision $R$ is itself an operational quantity: despite being expressed in terms of temporal moments, here we show that it can be inferred without access to an ideal time reference, and can be connected to experimentally meaningful clock-comparison observables. Having established this reference-free operational meaning of $R$, one can then ask how far the required resources can be reduced. In \cite{Dominguez2026Characterising}, we show that the auxiliary clock can in fact be removed altogether: a single clock can characterise its own precision using only a signal splitter and an imperfect delayed trigger, whose uncertainty is explicitly incorporated into the protocol. Moreover, this construction extends the framework beyond i.i.d. reset clocks to weakly dependent ticking processes, with reset clocks recovered as a special case. Thus, the present work establishes the operational foundation on which the single-clock, more general characterisation of \cite{Dominguez2026Characterising} is built.

%%%%%%%%%%%%%%%%%%%%%%%%%%%%%%%%%%%%%%%%%%%%%%%%%%%%%%%%%%%%%%%

\section{Theoretical framework}
\label{sec:theoframe}

% Definition of a independent ticking clock
Following the definition proposed in \cite{Woods_2021}, we conceptualise clocks as physical systems that autonomously emit temporal information to the environment in the form of ``ticks'', as depicted in \Cref{fig:ticks}. Consequently, the absence of an ideal time reference implies that the time interval between these ``ticks'' cannot be deterministic but must be characterised as a random variable $T$. While the most general definition of a clock imposes no constraints on the properties of these random variables, this paper focuses specifically on {\it independent ticking clocks}: 

Suppose the first $k-1$ ticks occur at times $0\leq \tau_1\leq \ldots\leq \tau_{k-1}$ respectively. Then, denoting by $P_k(t)$ the probability that the $k$-th tick occurs at time $t$, an independent ticking clock is one for which $P_k(t)$ is of the form $P_k(t)=P(t-\tau_k)$, for some $k$-independent probability $P:\mathbb{R}_{\geq 0}\to[0,1]$.

Independent ticking clocks have the defining property that they are characterised by the fact that they return to their initial state after generating each tick, such as clocks working on a limit cycle. As a result, each tick occurs independently of previous ticks, and all ticks follow the same probability distribution, making the ticking times $T_i$ independent and identically distributed (i.i.d.) variables. Although i.i.d. events represent a significant constraint for most random phenomena, this limitation is less restrictive for timekeeping devices, since most fundamental physical processes that prove useful for timekeeping, including atomic and nuclear state decays, naturally exhibit independent and identically distributed characteristics (indeed, our current definition of the second is based precisely on such an atomic transition). Therefore, a clock model featuring independent and identically distributed ticks has sound scientific justification \cite{Nurgalieva2024Tomography}.

One interesting property of independent ticking clocks is that, in the limit of very precise clocks or very large timescales, it is possible to directly relate their precision $R$ (used in theory-based articles \cite{Meier2025, Woods2022QuantumClocks, Erker2017AutonomousClocks, prech2024optimaltimeestimationclock}) to their Allan variance (the standard in experimental works),
\begin{equation}
    \label{eq:allanR}
    \sigma_A^2(\tau) = \frac{1}{R\ \mu\ \tau}.
\end{equation}
A proof of this result for quantum-mechanical clocks can be found in \cite{Silva2023TickingClocks}, and \Cref{sec:app_allan} gives a more general proof showing that it also holds for classical clocks.

\begin{figure}[htbp]
    \centering
    \includegraphics[width=0.5\textwidth]{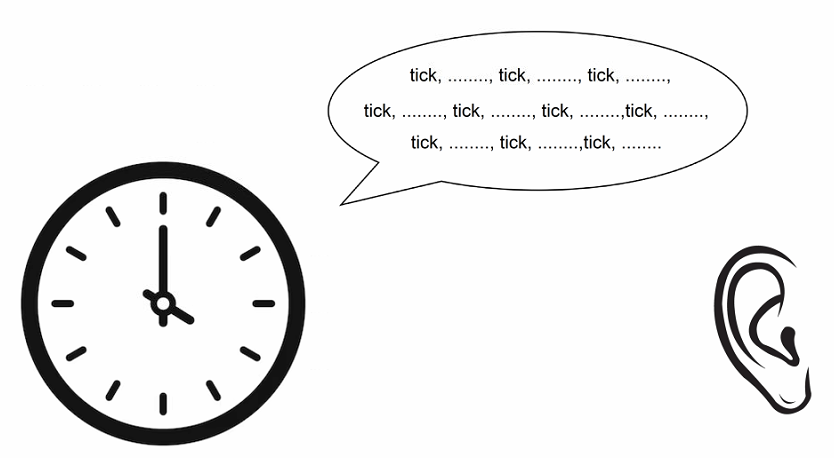}
    \caption{Normal behaviour of a clock: a physical system capable of transmitting information to the environment about the passage of time in the form of ``ticks'' in an autonomous manner.}
    \label{fig:ticks}
\end{figure}

% Introduction of the ticks as iids and picture with two clocks and lines
In this manuscript, we will work with experiments involving $N$ clocks with the same precision $R$. The behaviour of each copy is modelled as a discrete stochastic process $S_n = \sum_{i=1}^n T_i$, where $S_n$ represents the cumulative elapsed time required for a clock to produce $n$ ticks. By treating the individual tick intervals $T_i$ as i.i.d. random variables, we can use the formal framework of renewal and martingale theory. Specifically, this allows the application of Lorden's Inequality for renewal processes, Wald's First and Second Identities, and Doob's Optional Stopping Theorem to characterise the system's temporal evolution.

Regarding the figures of merit for clock precision, $R=\mu^2/\sigma^2$ is invariant across all copies of a particular clock, as they share identical $\mu$ and $\sigma$ values. Moreover, $R$ exhibits time-scale invariance: multiplying both $\mu$ and $\sigma$ by the same constant leaves $R$ unchanged. This property proves valuable in \Cref{sec:resultC}, as variations in a clock's oscillation frequency (whether due to manufacturing inconsistencies or external factors) do not affect $R$.

The alternative figure of merit in the literature~\cite{doubletick}, $\mathbb{E}[\tau]$, is defined through the following experimental procedure. Consider two identical copies of a clock, labelled A and B. We initiate copy A at time $t = 0$ and copy B at $t = \mu/2$. Initially, since both copies share the same mean ticking period, they should alternate their ticks. However, due to the inherent imperfection of real clocks, one can prove that, with probability $1$, eventually one copy will produce two consecutive ticks without an intervening tick from the other one. We define this event as a ``double tick'', and denote by $\tau$ the total number of ticks (from both copies combined) before this double tick occurs. A graphical depiction of this method is shown in \Cref{fig:defEt2}. Intuitively, more precise clocks yield higher values of $\mathbb{E}[\tau]$, and vice versa. This process can be generalised to any initial time offset $\Delta = \delta \mu$, which will become relevant for our analysis in \Cref{sec:results}.

\begin{figure*}[htbp]
    \centering
    \includegraphics[width=0.85\textwidth]{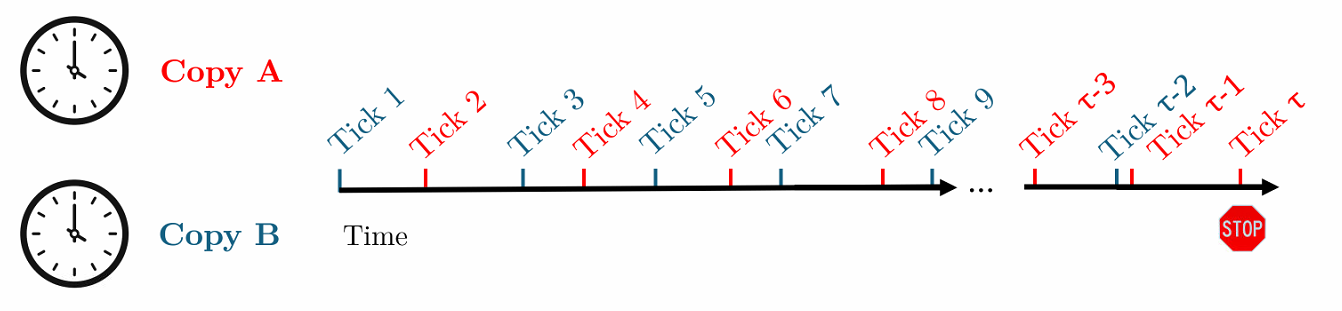}
    \caption{Experiment that defines the figure of merit $\mathbb{E}[\tau]$ as in \cite{doubletick}. The two copies of the clock are started with an initial delay $\Delta=\delta\mu$, and they tick freely in an alternating manner. Due to the stochastic nature of the ticking period, one will almost surely find a situation in which one of the two copies produces two consecutive ticks without the other copy performing any ticks: a \textit{double tick}. At that moment, we record the number of ticks that took place until then, denoted by $\tau$, and we stop and reset the experiment. Note that one does not need access to any external measuring device or reference frame: simply counting the number of ticks is enough to characterise the precision of the clock.}
    \label{fig:defEt2}
\end{figure*}

In the following sections, we address the physicality of $R$ (whether it is possible to measure a clock's $R$ value operationally without an external time reference) as well as the potential relationship between $R$ and $\mathbb{E}[\tau]$. Establishing a connection between these metrics would demonstrate the physicality of $R$, though finding such a relationship presents significant challenges.

To address these questions, \Cref{sec:results} presents various experiments that can be performed operationally to measure these figures of merit. Each experiment requires $N$ identical clock copies with an initial temporal offset of $\Delta=\delta\mu$ between them. In addition, we model scenarios involving multiple clocks with identical precision but different ticking frequencies, as previously discussed.

%%%%%%%%%%%%%%%%%%%%%%%%%%%%%%%%%%%%%%%%%%%%%%%%%%%%%%%%%%%%%

\section{Results}
\label{sec:results}
% TODO: Citas cruzadas entre las 2 publicaciones

The following subsections present different experiments that allow the characterisation of a timekeeping device by using a number $N$ of clocks with the same unknown precision $R$. For some of the characterisation protocols proposed, the clocks may start at different times, given an initial relative delay $\Delta$. Moreover, we also consider the possibility that they tick at different frequencies (and with correspondingly different means $\mu_1$ and $\mu_2$) for the same value of $R$, giving rise to a ratio $\alpha := \mu_{2}/\mu_{1}$ different from one.

\subsection{Many-equally-fast-clocks limit without initial offset ($\alpha = 1$, $N\to \infty$, $\Delta = 0$)}
\label{sec:resultA}

In this first case, we start with an arbitrarily large number of copies $N$ with the same ticking period ($\alpha = 1$). A single independent ticking clock that generates ticks over a given time interval can be regarded as a renewal process, which allows us to apply an argument based on the Central Limit Theorem for renewal processes and obtain a measure of the clock's precision.

To perform this experiment, each of the $N\in\mathbb{N}_{>0}$ copies is started at some arbitrary initial time. Importantly, there is no need to synchronise their starting times, as this would require an ideal time reference. At some random time after all clocks are running, we start counting the ticks each clock produces, and we stop counting at a later random time. We denote the time interval over which ticks are counted by the outcome of a random variable $\Theta_0 > 0$. Each experimental run has average duration $\mathbb{E}[\Theta_0]$. Let us denote by $\beta(x) \in\mathbb{N}_{>0}$ the random variable that counts the number of ticks a single copy would produce in time $x>0$. Since we are considering i.i.d. clocks, $\beta(\Theta_0)$ is identically distributed across copies. Intuitively, the more imprecise a copy is, the larger the standard deviation of $\beta(\Theta_0)$ is. Similarly, one expects larger $\mathbb{E}[\Theta_0]$ to increase the variance in the number of ticks, so there should be a direct relationship between $\mathbb{E}[\Theta_0]$, $R$ and the statistical moments of $\beta(\Theta_0)$, which we denote by

\begin{align}
\label{eq:def_case_A}
    \mu_{\beta(\Theta_0)} &:= \mathbb{E}[\beta(\Theta_0)],\\
    \nonumber \sigma_{\beta(\Theta_0)}^2 &:= \mathbb{E}[\beta(\Theta_0)^2] - \mathbb{E}[\beta(\Theta_0)]^2.
\end{align}

In the following, we assume that the variance of $\Theta_0$ increases at most sub-linearly with its mean, $\mathbb{E}[\Theta_0]$. This is a reasonable assumption, as discussed in \Cref{sec:app_subA}.

% definir de forma precisa mu y sigma, sobre qué se promedia

\begin{restatable}{theorem}{theoremone}
\label{th:1}

Let $\mu_{\beta(\Theta_0)}$ be the mean number of ticks recorded in the experiment and let $\sigma^2_{\beta(\Theta_0)}$ be the observed variance of this measurement as defined in \cref{eq:def_case_A}. Then, 

\begin{equation}
\label{eq:case_A}
	R = \frac{\mu_{\beta(\Theta_0)}}{\sigma^2_{\beta(\Theta_0)}} + O_\mathbb{P}\!\left(\frac{1}{\sqrt{N}}\right) + O\!\left(\frac{1}{\mathbb{E}[\Theta_0]}\right),
\end{equation}

as $N \to \infty,\ \mathbb{E}[\Theta_0] \to \infty$.

\end{restatable}

Here $O$ denotes big-$O$ notation and $O_\mathbb{P}$ indicates boundedness in probability.

The proof of this result is given in \Cref{sec:app_subA}. It is important to note that while both $\mu_{\beta(\Theta_0)}$ and $\sigma^2_{\beta(\Theta_0)}$ directly depend on $\Theta_0$, in the limit of large $N$ and $\mathbb{E}[\Theta_0]$, $\mu_{\beta(\Theta_0)}/\sigma^2_{\beta(\Theta_0)}$ does not.

Since \cref{eq:case_A} is derived from the application of the Central Limit Theorem, equality holds only in the asymptotic case $\mathbb{E}[\Theta_0]\to\infty$. Furthermore, since we estimate the statistical moments of $\beta(\Theta_0)$ by direct sampling over $N$ copies via the Strong Law of Large Numbers, the resulting sampling mean and variance estimates will be exact only in the limit $N\to\infty$.

\subsection{Two equal clocks with an arbitrary initial offset ($\alpha = 1$, $N = 2$, $\Delta \in [0,\mu)$)}
\label{sec:resultB}

In practice, manufacturing a large number $N$ of copies of a clock can be challenging. As an alternative, we study a second case where we start with only two copies of the clock ticking at the same period ($\alpha = 1$). The copies are started at $t=0$ and $t=\Delta$, respectively, with $0<\Delta<\mu$, and the procedure described in \Cref{sec:theoframe} to measure $\mathbb{E}[\tau]$ is performed. 

For the double-tick results below, we write
\[
    T_i' := T_i-\mu
\]
for the centred tick-time fluctuation of either copy. We assume throughout the paper that, for some fixed $p>2$, the centred ticking distribution has uniformly bounded $p$-th absolute moment in the precision limit:
\begin{equation}
    \label{eq:moment_assumption}
    \mathbb{E}\!\left[\left|\frac{T_i'}{\sigma}\right|^p\right]\le C_p ,
\end{equation}
where $C_p$ is independent of $R$, $\mu$ and $\delta$. Broadly speaking, \cref{eq:moment_assumption} ensures that as our clock gets more precise ($R$ increases), the tails of the distribution will not grow indefinitely (the probability of producing a tick far from the mean gets smaller the more precise a clock is). An example of this would be clocks whose ticking distributions have exponentially decaying tails, such as Gaussian or Poissonian distributions. The constants appearing in the following bounds may depend on $p$ and $C_p$, but not on the clock precision.
We also write
\[
    P_{even}:=P[\tau \text{ even}],\qquad P_{odd}:=P[\tau \text{ odd}]
\]
for the probabilities that the double tick takes place at an even or odd number of ticks, respectively.

\begin{restatable}{corollary}{theoremtwo}
\label{th:2}
For every clock, one can always find a $\kappa$ such that

\begin{equation}
\label{eq:case_B_eq1}
    \left| \frac{\mathbb{E}[\tau]}{R} - \delta(1-\delta) \right| \le \frac{\kappa}{R^{1/2-1/p}}
\end{equation}

for any sufficiently large $R$, where $\delta=\Delta/\mu$ is the normalised initial time offset, satisfying $0<\delta<1$.
%\mpwb{I will wait until the proof is correct to rewrite this}
\end{restatable}

The proof of this result, which makes use of martingale theory and Doob's Optional Stopping Theorem, is provided in \Cref{sec:app_subB}.

\Cref{eq:case_B_eq1} depends on the normalised initial offset $\delta=\Delta/\mu$. At first glance this suggests that an accurate estimate of $\delta$ is required to determine $R$, which would reintroduce dependence on an external timing reference. \Cref{th:3} shows this is not the case: $\delta$ can be estimated from the same double-tick data with an error bound that vanishes as $R \to\infty$. 

\begin{restatable}{theorem}{theoremthree}
\label{th:3}
The normalised initial time offset $0<\delta<1$ satisfies

\begin{equation}
    \label{eq:case_B_eq3}
    \left|\delta - P_{odd}\right| \le \frac{\kappa_{p}}{R^{1/2-1/p}}
\end{equation}

for some $\kappa_p > 0$ that is independent of $R$ and $\delta$.

%\mpwb{I will wait until the proof is correct to rewrite this}
\end{restatable}

The proof (detailed in \Cref{sec:app_subB}) again uses martingale theory and is based on constructing suitable martingales to which one can apply the Optional Stopping Theorem.

As discussed in \Cref{sec:discussion}, this result shows not only that our two figures of merit for the precision of a clock are proportional in the limit of highly precise clocks, thereby providing a method for operationally determining $R$, but also that we can measure time intervals $\Delta$ smaller than the ticking period $\mu$, namely $\Delta < \mu$. Nevertheless, if one intended only to measure $R$, and since the derivative with respect to $\delta$ of $\mathbb{E}[\tau]/R$ has a minimum at $\delta = 1/2$, this is the initial offset one should aim for in order to minimise the uncertainty in the measurement of the precision of the clock.

Having a fixed constant $\delta$ may be unrealistic from an experimental point of view (due to, for example, interactions of the setup with the environment). In this case, it is worth noting that one can also decide not to estimate or keep track of the initial time shift between the copies and treat $\delta$ as a random variable instead: this will still yield a valid approach to measure $R$ (for example, a uniform distribution in $\delta$ would result in $\mathbb{E}[\tau] / R = 1/6 + O\left(R^{-(1/2-1/p)}\right)$).

With this in mind,~\Cref{th:2} can be further refined into a bound which is unconditional on $\delta$:

\begin{restatable}{theorem}{theoremfour}
\label{th:4}

The following unconditional bound for $R$ holds:
%\iffalse
\begin{equation}
\label{eq:case_B_eq2}
    0 \le \frac{\mathbb{E}[\tau]}{R} - P_{odd}P_{even} \le \frac{\kappa_{1,p}}{R^{1/2-1/p}}+\frac{\kappa_{2,p}}{R^{1-2/p}}.
\end{equation}
%\fi
%\mpwb{I will wait until the proof is correct to rewrite this}
\end{restatable}

Here $\kappa_{1,p} > 0$ and $\kappa_{2,p} > 0$ do not depend on $R$ or $\delta$. In the limit of large $R$, \cref{eq:case_B_eq2} is equivalent to

\begin{equation}
    \frac{\mathbb{E}[\tau]}{R} = P_{odd}P_{even} + O\left(\frac{1}{R^{1/2-1/p}}\right).
\end{equation}

Importantly, all of the variables appearing in the upper and lower bounds can be experimentally determined to arbitrary precision, thereby allowing the precision of a clock $R$ to be characterised, as well as allowing measurements of times shorter than $\mu$, both with a bounded uncertainty that decreases as $R^{-(1/2-1/p)}$. The proof of \Cref{th:4} can be found in \Cref{sec:app_proofs}.

Finally, the martingale method gives a hierarchy of estimators for arbitrary non-degenerate initial offsets. The following result generalises \Cref{th:2} to higher moments of $\tau$. To state it precisely, write $\delta:=\Delta/\mu$ and assume $0<\delta<1$.

\begin{restatable}{theorem}{theoremfive}
\label{th:5}
Fix $m\in\mathbb{N}_{>0}$. Assume that the centred ticking distribution $T'_i$ satisfies \cref{eq:moment_assumption} for some $p>2m$, and that $0<\delta<1$. Let $P_0(\delta):=1$. For $j\ge1$, let $P_j$ be the unique polynomial of degree $2j$ satisfying
\begin{align}
    P_j(0)&=P_j(1)=0,\\
    \nonumber\frac{1}{2}\frac{d^2}{d \delta^2} P_j(\delta)&=-jP_{j-1}(\delta),
    \qquad 0<\delta<1 .
\end{align}
Then there exists a coefficient $\kappa_{m,p}>0$, independent of $R$ and $\delta$, such that, for all sufficiently large $R$,
\begin{equation}
    \left|
    \frac{\mathbb{E}[\tau^m]}{R^m}
    -P_m(\delta)
    \right| \le
    \frac{\kappa_{m,p}}{R^{1/2-1/p}}.
\end{equation}

\end{restatable}

A proof of \Cref{th:5} can be found in \Cref{sec:app_subB}.

\subsection{Two equally precise clocks (same $R$) with different ticking frequencies and arbitrary initial offset ($\alpha \neq 1$, $N = 2$, $\Delta \in [0,\mu)$)} 
\label{sec:resultC}

We also study a more complex case in which the two copies are not required to tick at the same period $\mu$. This can correspond, for example, to the real-world scenario in which, due to manufacturing imperfections, the oscillation period of the clockwork differs from one copy to the other, while preserving the accuracy (i.e., preserving the ratio $\mu^2/\sigma^2$). Another situation in which this case can occur in reality is for two equal copies of a clock such that one experiences time dilation with respect to the other (potentially due to a non-inertial frame of reference or the action of gravity). 

Since the equations that lead to this solution are quite lengthy, they have been placed in \Cref{sec:app_eqC} rather than interrupting the main text. Readers interested in their complete mathematical expressions can refer to that section. The solutions in \Cref{sec:app_eqC} were obtained only for the limit of highly precise clocks ($R\to\infty$) whose ticking distribution approaches a so-called mesokurtic one ($\mathbb{E}[T_i'^4]/\sigma^4 - 3 = o(1)$, where $o(\cdot)$ stands for little-o notation). However, the methods shown in \Cref{sec:app_subB} may yield more general bounds, although the calculations quickly become more involved.

\Cref{sec:app_eqC} provides closed-form expressions from which one can simultaneously determine the precision $R$, the initial delay $\delta$ and the frequency correction $\alpha$ from experimentally measurable variables by using only two clocks of equal precision. In practical scenarios, this setup will not only allow the characterisation of a clock or the measurement of time intervals shorter than its resolution $\mu$, but will also enable the detection and measurement of frequency offsets, time-dilation factors \cite{Bothwell2022, Ashby2003, HafeleKeating1972}, Doppler shifts \cite{Delva2017} and similar effects.

In this more complex case, three main variables determine when the double tick is triggered, and hence the value of $\mathbb{E}[\tau]$. First, there is the precision $R$, so that more precise clocks naturally exhibit larger $\mathbb{E}[\tau]$. Second, the initial delay $\delta$ also affects $\mathbb{E}[\tau]$: the closer $\delta$ is to $1/2$, the larger $\mathbb{E}[\tau]$ can become. Finally, the difference in ticking period given by $\alpha$ modulates the mean number of ticks before the double tick, maximising it when $\alpha$ is close to $1$. In an experiment where the two copies are started in counter-phase ($\delta \approx 1/2$), $\mathbb{E}[\tau]$ is determined by the interplay between $R$ and $\alpha$. As in the previous case, to characterise $R$ with high accuracy the two copies must satisfy $\alpha \approx 1$. However, now the converse is also true: one can measure frequency offsets or time-dilation factors with arbitrary precision, given sufficiently precise clocks.

%%%%%%%%%%%%%%%%%%%%%%%%%%%%%%%%%%%%%%%%%%%%%%%%%%%%%%%%%%

\section{Example: a Gaussian clock}

%Definition of the simulation procedure and parameters chosen.
In this section, a numerical Monte Carlo simulation of the double-ticking clock experiment was performed to test the results from \Cref{sec:results}. To do this, a script generated random tick lengths one at a time, alternating between the clocks, and stopped only when it detected a double tick. The probability distribution of the ticks was chosen to be normal with a normalised mean of $\mu = 1$. 

\begin{figure}[htbp]
    \centering
    \includegraphics[width=0.45\textwidth]{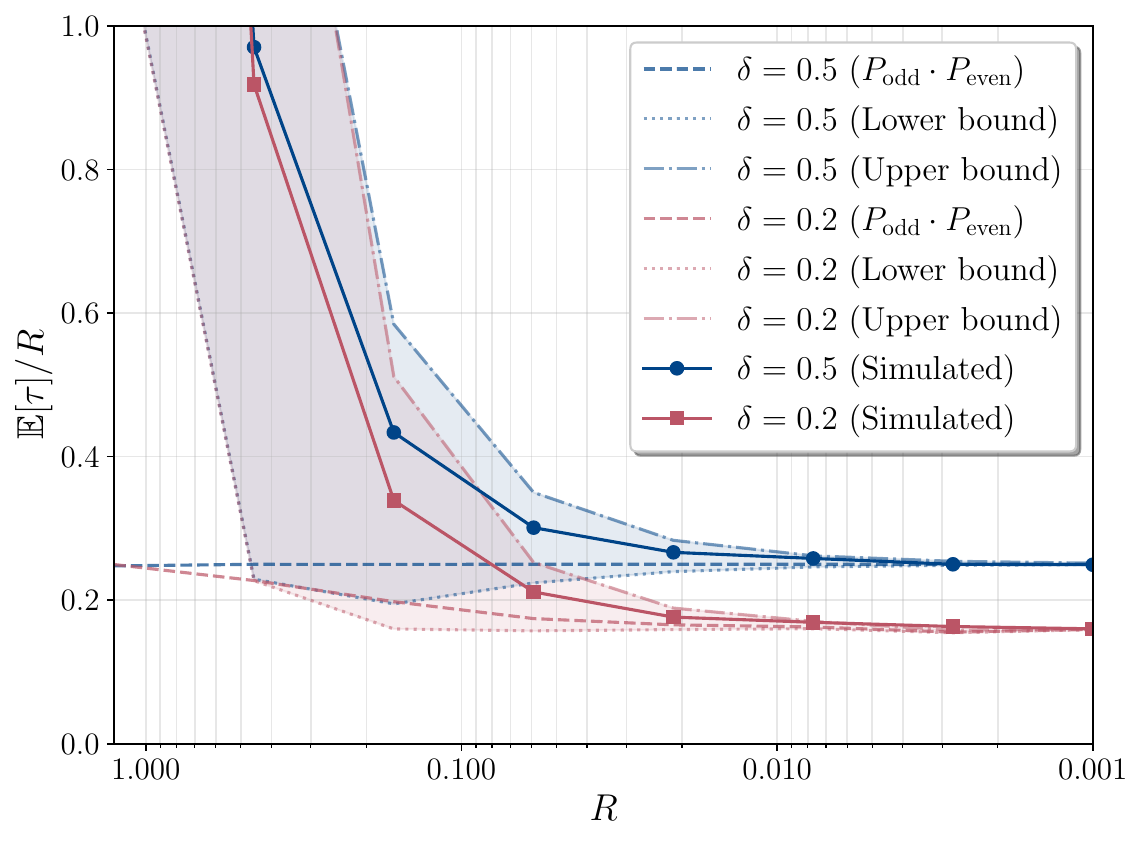}
    \caption{Monte Carlo simulation ($10000$ iterations performed) of two clocks following a normal ticking distribution with $\mu = 1$ and $R$ ranging from $10^2$ to $10^4$, both for $\delta = 0.5$ (blue) and $\delta = 0.2$ (salmon). The dotted (dash-dotted) line represents the lower (upper) bound from \Cref{sec:resultB}, and the dashed line shows the asymptotic limit of $\mathbb{E}[\tau]/R$ as measured in the experiment. Dots and straight lines represent the results from the actual simulations.}
    \label{fig:expBSimu}
\end{figure}

To test \Cref{sec:resultB}, a parametric sweep for $R$ was performed for two different values of $\delta$ ($\delta = 0.5$ and $\delta = 0.2$). Similarly, a parametric sweep for $\delta$ ranging from $0.1$ to $0.9$ was simulated for a fixed $R=10^4$, as shown in \Cref{fig:expBDeltaSimu}.

\begin{figure}[htbp]
    \centering
    \includegraphics[width=0.45\textwidth]{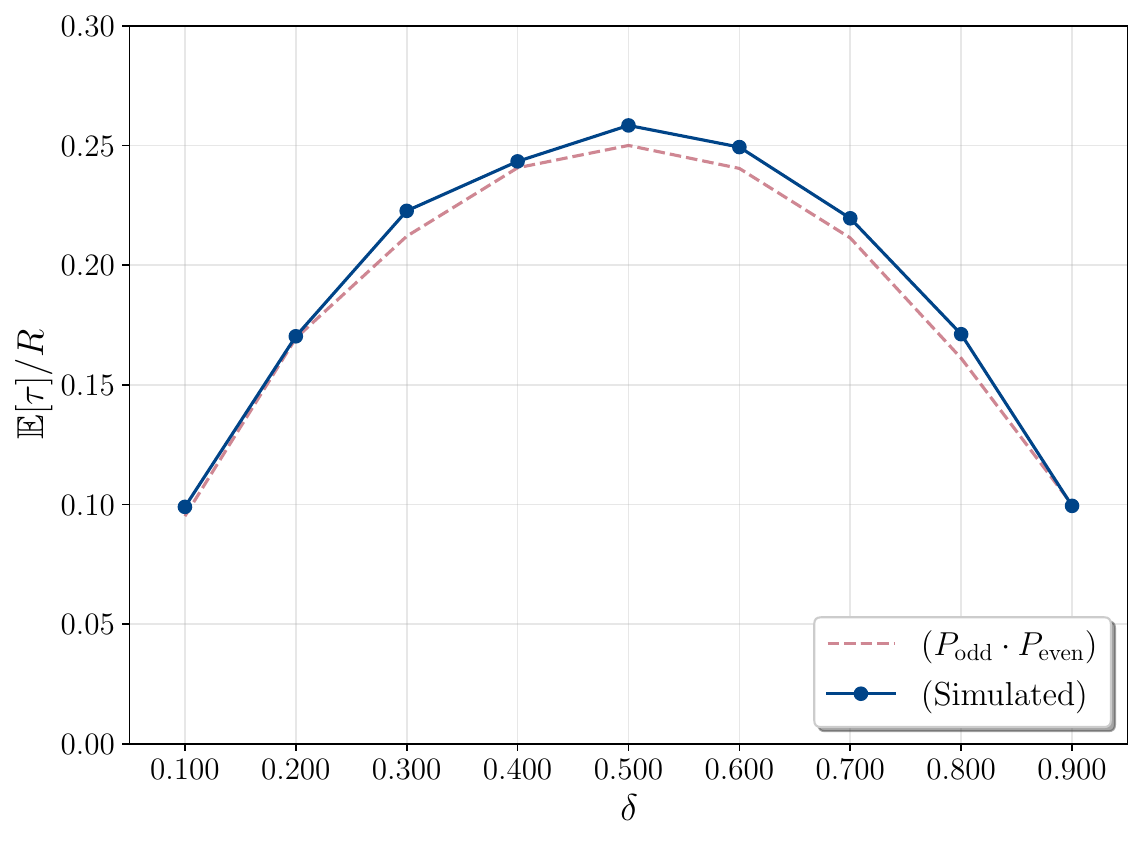}
    \caption{Clocks following a normal ticking distribution with $\mu = 1$ and $R=10^4$, for different values of $\delta$ ($10000$ iterations performed). The blue line represents the measured values, and the salmon dashed line the theoretically derived asymptotic bound $P_{odd}P_{even}$. The distribution follows a semicircle law.}
    \label{fig:expBDeltaSimu}
\end{figure}

The generated plots were obtained using only information available in a realistic experimental scenario, thereby allowing the experimentalist to estimate only the probability $P_{odd}$ and the mean value $\mathbb{E}[\tau]$. A clear asymptotic trend towards the value predicted by our results is observed, thus supporting the argument that the precision $R$ of a clock can be measured operationally.

%%%%%%%%%%%%%%%%%%%%%%%%%%%%%%%%%%%%%%%%%%%%%%%%%%%%%%%%%%%%%%%%%%%%%%%%%%%%%%%%%%

\section{Discussion and outlook}
\label{sec:discussion}

% $R$ is physical.
This work has established several important results in clock characterisation without the need for an external time reference. The theoretically defined figure of merit for the precision of a clock, $R=\mu^2/\sigma^2$, is shown to be physically meaningful through our experimental approaches. Our results show that $R=\mu^2/\sigma^2$ is operationally accessible under the stated i.i.d. clock model. The bounds in \Cref{sec:results} imply that the estimation error can be made arbitrarily small in the asymptotic regime (large $R$ and/or many independent trials). We therefore provide an operational route to estimate $R$, with explicit finite-sample error terms given in \Cref{th:1,th:3,th:4}.

% It is posible to (asymptotically) relate $R$ to $\mathbb{E}[\tau]$
In this regard, we have demonstrated a clear asymptotic relation between the theoretical parameter $R$ and the experimentally accessible $\mathbb{E}[\tau]$, as well as with the well-known Allan variance. These relationships are particularly valuable since $\mathbb{E}[\tau]$ can be directly measured through the double-tick experiment described in \Cref{sec:theoframe}, while $R$ has traditionally been considered a theoretical construct requiring an ideal time reference. The proportionality between these figures of merit enables a practical bridge between theory and experiment, validating both the physicality of $R$ and the utility of $\mathbb{E}[\tau]$ as meaningful indicators of clock precision. Moreover, the direct relationship between $R$ and the Allan variance for independent ticking clocks creates a powerful link between the current literature in experimental and theoretical studies of timekeeping devices.

% The proposed experiments can be used not only to characterise a clock but also to measure timescales smaller than their ticking time with arbitrary precision (by measuring $\delta$) or time dilation factor (by measuring $\alpha$).
The methods developed in this paper extend beyond mere clock characterisation to practical metrological applications. By measuring $\delta$, we can determine time intervals much smaller than the clock's natural ticking period ($\Delta \ll \mu$) with a precision that improves as $R^{-(1/2-1/p)}$ under \cref{eq:moment_assumption}. This capability has profound implications for high-precision timing applications where sub-period measurements are critical. Similarly, our approach enables direct measurement of time dilation factors (through $\alpha$) without relying on absolute time references, offering new opportunities for experimental tests of relativistic effects in scenarios where gravitational fields or non-inertial reference frames introduce frequency shifts between otherwise identical clocks \cite{Bothwell2022, Ashby2003, HafeleKeating1972}.

% Drawbacks of the current method: convergence may be slow for incredibly precise clocks (it could be made faster by using more than $N = 2$ nevertheless) and, from an experimental point of view, manufacturing two copies of the same clock with equal $R$ may be challenging.
There are some practical considerations that should be taken into account regarding the experimental implementation of the methods showcased in \Cref{sec:results}. For example, in the case of extremely precise clocks (large $R$ values), convergence to the asymptotic relationships presented may be slow, potentially requiring many experimental runs to achieve the desired precision. This challenge could be partially mitigated by employing more than $N = 2$ copies in synchronised configurations, though at the cost of increased experimental complexity. Moreover, from a practical standpoint, manufacturing multiple copies of a clock with identical $R$ values presents significant engineering challenges, particularly for high-precision devices. Small manufacturing variations could introduce systematic errors that must be carefully accounted for in experimental implementations, since they can lead to two non-identical copies of the clock with slightly different $R$.

Despite these challenges, the framework presented here provides a robust foundation for clock characterisation that aligns with the fundamental constraints imposed by quantum theory. By eliminating the theoretical requirement for an unattainable ideal time reference, these methods offer both conceptual clarity and practical utility for advancing the science of chronometry.

%%%%%%%%%%%%%%%%%%%%%%%%%%%%%%%%%%%%%%%%%%%%%%%%%%%%%%%%%%%%%%%%%%%%%%%%%%%%%%%

\section{Conclusion}

This study addresses a fundamental challenge in quantum metrology: determining clock precision without relying on an external time reference. We have demonstrated that multiple non-ideal clocks can collectively characterise their own precision through various experimental approaches without requiring an external reference frame.

Our methods establish a clear relationship between the theoretical precision parameter $R = \mu^2/\sigma^2$ and experimentally accessible measurements such as $\mathbb{E}[\tau]$ or the Allan variance, validating the physicality of $R$ and providing operational means to measure it with bounded uncertainty and arbitrary precision. This relationship creates a valuable bridge between theoretical models and experimental implementations of timekeeping devices.

Beyond clock characterisation, our framework enables precise measurements of sub-period time intervals, frequency offsets, time-dilation factors and Doppler shifts each with an uncertainty that decreases as $R^{-(1/2-1/p)}$ under the stated moment assumptions. These capabilities have significant implications for quantum metrology and fundamental physics, offering new approaches to temporal measurement that align with the constraints imposed by quantum theory.

While practical implementation challenges exist, these methods provide a robust foundation for advancing chronometry within the fundamental limitations of quantum mechanics, eliminating the theoretical requirement for an unattainable ideal time reference.

\begin{acknowledgments}
This work was supported by the MSCA Cofund QuanG (Grant Number: 101081458), funded by the European Union. The views and opinions expressed are those of the author(s) only and do not necessarily reflect those of the European Union. Neither the European Union nor the granting authority can be held responsible for them.

All equations and ideas, except for \Cref{th:5}, were derived by hand. A simpler version of \Cref{th:5} was also derived by hand, but its final generalised form was proved in an AI-assisted manner using Claude, Qwen and ChatGPT.

\iffalse
\begin{figure}[htbp]
    \centering
    \begin{minipage}{0.12\textwidth}
        \centering
        \includegraphics[width=\textwidth]{img_files/logo_EU.jpg}
    \end{minipage}
    \hspace*{0.08\textwidth} % Adds horizontal space between images
    \begin{minipage}{0.12\textwidth}
        \centering
        \includegraphics[width=\textwidth]{img_files/logo-quang.png}
    \end{minipage}
\end{figure}
\fi

\end{acknowledgments}

\appendix
% Ensure cleveref identifies appendix sections by their appendix designation
% rather than referring to them as ordinary sections.
\crefname{appendix}{appendix}{appendices}
\Crefname{appendix}{Appendix}{Appendices}
\crefalias{section}{appendix}
\crefalias{subsection}{appendix}
\renewcommand{\thesubsection}{\thesection.\Alph{subsection}}

%TODO
%\mpwb{TODO: repositorio en zenodo}
%\mpwb{TODO: Revisar con Wolfram}

\section{Equations for the case $N=2$, $\Delta \in [0,\mu)$, $\alpha > 0$}
\label{sec:app_eqC}

\subsection{Full expressions for the solutions to \Cref{sec:resultC}}

Our precision measure $R$ is given by 

\begin{widetext}
    \begin{equation}
    \label{eq:long_3}
        R = \frac{2(1+\alpha^2(\mathbb{E}[\tau]-1)+\mathbb{E}[\tau])}{4\delta^2+2\mathbb{E}[\tau]\delta(\alpha-1) + (\alpha-1)^2\mathbb{E}[{\tau }^2] 
        - P_{even}(4(\alpha+1)\delta +1+ \mathbb{E}[\tau]+\alpha(2+\alpha-\alpha \mathbb{E}[\tau]))} + o(R),
    \end{equation}
    which is a function of the experimentally measurable parameters $\mathbb{E}[\tau]$, $\mathbb{E}[\tau ^2]$ and $P_{even}$, with $P_{odd}:=1-P_{even}$, and the relevant quantities $\alpha$ and $\delta$ (i.e. the ratio between the mean ticking period of the two clocks and the normalised initial delay between them). The $\alpha$ parameter is expressed solely as a function of $\delta$ and these experimentally measurable parameters, where $\alpha$ is given by 
    \begin{equation}
    \label{eq:long_2}
        \alpha = \frac{\mathbb{E}[\tau]+P_{even}-2\delta}{\mathbb{E}[\tau]-P_{even}} + o(1)
    \end{equation}
    while $\delta$ is the solution to the following equation once substituting in for $R$ and $\alpha$ via~\cref{eq:long_3,eq:long_2}.
    \begin{align}
        \begin{split}
            &(2\delta+(\alpha-1)\mathbb{E}[\tau])^4P_{odd}R^2 + (1+\alpha-2\delta+\mathbb{E}[\tau]-\alpha\mathbb{E}[\tau])^4P_{even}R^2+48\mathbb{E}[{\tau }^2]\\
            &= 24\mathbb{E}[\tau]R((2\delta + (\alpha-1)\mathbb{E}[\tau])^2- (1+\alpha)(4\delta+2(\alpha-1)\mathbb{E}[\tau]-\alpha-1)P_{even}) + o(R^2),
        \end{split}\label{eq:long_1}
    \end{align}
\end{widetext}

%\mpwb{NOTA: se puede derivar un resultado usando O mayúscula?}

where the little o notation in~\cref{eq:long_3,eq:long_2,eq:long_1} refers to the limit of large $R$. These three equations allow one to solve for the three unknowns ($R$, $\alpha$, $\delta$) and obtain a full characterisation of our clocks in terms of variables one can measure operationally ($\mathbb{E}[\tau]$, $\mathbb{E}[\tau ^2]$, $P_{even}$).

While these equations hold in the limit of large $R$, it may be possible to obtain bounds analogous to the ones presented in \cref{eq:case_B_eq2} by following a similar procedure to that proposed therein.

\subsection{Proof of \Cref{eq:long_1,eq:long_2,eq:long_3}}
\label{sec:app_subC}

In the case where one of the copies is a dimensionless factor $\alpha$ times faster than the other, the mathematical definition of the problem varies slightly. First, note that both copies will still maintain the same $R$ (the $\alpha^2$ factors cancel out in its definition). However, let, as in the main text, the superscript $(A)$ or $(B)$ indicate the copy of the clock the expression refers to (the superscript will be omitted in expressions that are copy-independent such as $T_i'$). Then, since $\mathbb{E}[T_i^{(A)}] := \mu^{(A)} \neq \alpha \mu^{(A)}  = \mu^{(B)} =: \mathbb{E}[T_i^{(B)}]$, simply grouping the ticks from both copies into a single summation would lead to a sum of non-identically distributed variables. As a result, extra care must be taken when working with the mean-zero random variables $T_i' := T_i-\mathbb{E}[T_i]$.

With this in mind, one can restate the problem in terms of $\alpha$ as follows:

\begin{equation}
\label{eq:app_C1}
    \text{Double tick} \iff
    \begin{cases}
      \sum_{i = 1}^n (-1)^{i+1}T_i' + d_n
      < \Delta - \mu^{(A)}\frac{1+\alpha}{2},\\
      \sum_{i = 1}^n (-1)^{i+1}T_i' + d_n
      > \Delta.
    \end{cases}  
\end{equation}

Here $d_n := (1-\alpha)(n/2)\mu^{(A)}$ represents a drift-generated delay, and the first (second) case corresponds to copy $A$ ($B$) producing the double tick. We can now define $S_{\alpha, n} := \sum_{i = 1}^n (-1)^{i+1}T_i'$ and the stopping time $\tau_\alpha$ by

\begin{align}
\label{eq:app_C2}
    \nonumber \tau_{\alpha} :=\min\biggl\{n: & \left(n \text{ odd} \land S_{\alpha, n}<\Delta- \mu^{(A)}\frac{1+\alpha}{2}\right)\\ 
    & \lor \left(n \text{ even} \land S_{\alpha, n}>\Delta\right)\biggr\}.
\end{align}

With these definitions, one finds that the following is a martingale:

\begin{equation}
\label{eq:app_C3}
    S_{\alpha,n}^2 - n\ \frac{{\sigma^{(A)}}^2 + {\sigma^{(B)}}^2}{2} - \frac{{\sigma^{(A)}}^2 - {\sigma^{(B)}}^2}{2},
\end{equation}

and, under the assumption of a sufficiently well-behaved ticking distribution, applying the Optional Stopping Theorem to \cref{eq:app_C3} leads to:

\begin{align}
\label{eq:app_C4}
    \nonumber {\sigma^{(A)}}^2\mathbb{E}[\tau_{\alpha}] & = \mathbb{E}[S_{\alpha, \tau}^2|\tau_{\alpha} \text{ even}]\ P[\tau_{\alpha} \text{ even}]\\
    & +  \mathbb{E}[S_{\alpha, \tau}^2|\tau_{\alpha} \text{ odd}]\ P[\tau_{\alpha} \text{ odd}].
\end{align}

Using the same technique as in \Cref{sec:app_subB}, one arrives at the bounds in \cref{eq:app_C5,eq:app_C6}. 

\begin{widetext}
\begin{equation}
\label{eq:app_C5}
\begin{aligned}
    &\mathbb{E}[S_{\alpha, \tau}^2] < P_{odd} \biggr[ \left( \Delta - \mu \frac{1+\alpha}{2}\right)^2 + \frac{\mathbb{E}[{\tau_\alpha}^2]}{4}\mu^2(1-\alpha)^2 - 2 \left( \Delta - \mu\frac{1+\alpha}{2}\right)\mu(1-\alpha)\frac{\mathbb{E}[\tau_\alpha]}{2} \biggr] \\
    +P_{even}  \biggr[ \Delta^2 +& \frac{\mathbb{E}[{\tau_\alpha}^2]}{4}\mu^2(1-\alpha)^2 + \mu^2\frac{(1+\alpha)^2}{4} - \frac{\mathbb{E}[\tau_\alpha]}{2}\mu^2(1+\alpha)^2
    -2\Delta\frac{\mathbb{E}[\tau_\alpha]-1}{2}\mu^2(1+\alpha)^2 + \sigma^2 + 2\mathbb{E}[S_{\alpha, \tau-1}^2T_\tau'] \biggr]
\end{aligned}
\end{equation}
%\todo{revisar si el cuadrado va dentro de la media}

The complementary bound is

\begin{equation}
\label{eq:app_C6}
\begin{aligned}
    &\mathbb{E}[S_{\alpha, \tau}^2] > P_{odd}\biggr[ \left(\Delta - \mu\frac{1+\alpha}{2}\right)^2 + \frac{\mathbb{E}[{\tau_\alpha}^2]}{4}\mu^2(1-\alpha)^2 + \frac{\mu^2(1 - \alpha)^2}{4} - \frac{\mathbb{E}[\tau_\alpha]}{2}\mu^2(1+\alpha)^2 \\
    - 2 &\left( \Delta - \mu\frac{1+\alpha}{2} \right) \mu(1-\alpha)\frac{\mathbb{E}[\tau_\alpha-1]}{2} + \sigma^2 + 2\mathbb{E}[S_{\alpha, \tau-1}T_\tau'] \biggr] + P_{even}\biggr[ \Delta^2 + \frac{\mathbb{E}[{\tau_\alpha}^2]}{4}\mu(1-\alpha)^2 - 2\Delta\mu(1-\alpha)\frac{\mathbb{E}[\tau_\alpha]}{2} \biggr].
\end{aligned}
\end{equation}
\end{widetext}

Finally, by letting $R \to \infty$, both bounds converge to the expression \cref{eq:app_C7}, where we have already substituted $\Delta = \delta\mu$. 

\begin{widetext}
\begin{align}
\label{eq:app_C7}
    \nonumber \frac{1}{\mu^2}\ \mathbb{E}[S_{\alpha, \tau}^2] = &P_{odd}\biggr[ \left( \delta - \frac{1+\alpha}{2}\right)^2 + \frac{\mathbb{E}[{\tau_\alpha}^2]}{4} (1-\alpha)^2 -
    \frac{\mathbb{E}[\tau_\alpha]}{2}(1-\alpha)\left( \delta - \frac{1+\alpha}{2}\right)\biggr] + \\
    &P_{even}\biggr[\delta^2 + \frac{\mathbb{E}[{\tau_\alpha}^2]}{4} (1-\alpha)^2 - \frac{\mathbb{E}[\tau_\alpha]}{2}(1-\alpha)\delta \biggr] + o(1),
\end{align}
\end{widetext}

where $o(1)$ is little-o notation, meaning that equality holds in the limit $R \to \infty$. Substituting \cref{eq:app_C7} into \cref{eq:app_C4} yields an equation in three unknowns: $R$, $\delta$ and $\alpha$. One may find another such equation by using the Optional Stopping Theorem directly on $S_{\alpha, \tau}$, leading to $S_{\alpha, \tau} = S_{\alpha, 0} = 0$, which, in the limit of very precise clocks, simplifies to:

\begin{equation}
\label{eq:app_C8}
\begin{aligned}
    &P_{even} \biggr[ \delta - \frac{1+\alpha}{2}
    - \frac{\mathbb{E}[\tau_\alpha]}{2}(1-\alpha) \biggr]\\
    &\quad + P_{odd}\biggr[\delta-\frac{\mathbb{E}[\tau_\alpha]}{2}(1-\alpha) \biggr]
    = o(1),
\end{aligned}
\end{equation}

Solving this equation for $\alpha$ directly recovers \cref{eq:long_2}.

Finally, one may find the last necessary equation by applying the Optional Stopping Theorem in a similar fashion to \cref{eq:app_B12}, using an equivalent martingale obtained by substituting $\tau\to\tau_\alpha,\ S_n \to S_{\alpha,\tau}$. If one further assumes that the ticking distribution of the clock resembles a normal distribution as it becomes more accurate, that is, $\mathbb{E}[T_i'^4]/\sigma^4 - 3 = o(1)$, then in the limit $R\to\infty$, \cref{eq:long_3} is recovered.

\qed

\section{Proof of the main results}
\label{sec:app_proofs}

\subsection{Proof of \Cref{th:1}}
\label{sec:app_subA}

For convenience, we restate here \Cref{th:1}:

\theoremone*

\begin{proof}

We first consider the case of a single clock. Let $\vartheta_0$ denote the deterministic duration for which the experiment is run (i.e., a realisation of $\Theta_0$), and let $T_i$ be the duration of the $i$th inter-tick interval as defined in the main text. Since we are working with independent ticking clocks, $\{T_i\}_i$ are independent and identically distributed variables. Now, let the tick-counting function $\beta(t)$ be defined as follows:

%quitar teta 0 como media y dejar solol la media

\begin{equation}
\label{eq:app_A1}
    \beta(\vartheta_0) := \sup\biggl\{n\ :\sum_{i=1}^n T_i \le \vartheta_0\biggr\}.
\end{equation}

$\beta(\vartheta_0)$ is then a renewal process. Standard renewal-theory asymptotics (under a finite second moment of $T_i$) and an application of the Central Limit Theorem for this type of process imply \cite{andersen1987central, reinert2018bound, CLTrenew2}:

\begin{equation}
\label{eq:app_A2}
    \frac{\beta(\vartheta_0) - \vartheta_0/\mu}{\sqrt{\vartheta_0\sigma^2/\mu^3}} \to \mathcal{N}(0,1)\ \text{ as }\ \vartheta_0\to \infty.
\end{equation}

Equivalently, for the first two moments,

\begin{equation}
\label{eq:app_A3}
    \mathbb{E}[\beta(\vartheta_0)] = \frac{\vartheta_0}{\mu} + O(1),\qquad
    \operatorname{Var}(\beta(\vartheta_0)) = \frac{\sigma^2}{\mu^3}\vartheta_0 + O(1),
\end{equation}

so that, rearranging terms,

\begin{equation}
\label{eq:app_A4}
\frac{\mathbb{E}[\beta(\vartheta_0)]}{\operatorname{Var}(\beta(\vartheta_0))}
= R + O\!\left(\frac{1}{\vartheta_0}\right).
\end{equation}

We must now study how the result changes if we allow the duration of the experiment (i.e. $\vartheta_0$) to be a random time $\Theta_0$. On the one hand, by the law of total expectation, $\mathbb{E}[\beta(\Theta_0)] = \mathbb{E}[\mathbb{E}[\beta(\vartheta_0)|\Theta_0]]$. On the other hand, by the law of total variance, $\operatorname{Var}(\beta(\Theta_0)) = \mathbb{E}[\operatorname{Var}(\beta(\vartheta_0))|\Theta_0] + \operatorname{Var}(\mathbb{E}[\beta(\vartheta_0)|\Theta_0])$. Substituting these into \cref{eq:app_A4} gives, in the limit of large $\mathbb{E}[\Theta_0]$,

\begin{equation}
\label{eq:app_A5}
\operatorname{Var}(\beta(\Theta_0)) = \frac{\sigma^2}{\mu^2}\mathbb{E}[\beta(\Theta_0)] + \frac{\operatorname{Var}(\Theta_0)}{\mu^2} + O(1).
\end{equation}

One can recover an expression equivalent to \cref{eq:app_A4} by dividing both sides by $\mathbb{E}[\beta(\Theta_0)]$, as long as $\operatorname{Var}(\Theta_0)$ grows sub-linearly with respect to $\mathbb{E}[\Theta_0]$. In particular, if $\operatorname{Var}(\Theta_0)=O(1)$ (small timing tolerance), then
\begin{equation}
\label{eq:app_A5b}
\frac{\mathbb{E}[\beta(\Theta_0)]}{\operatorname{Var}(\beta(\Theta_0))}
= R + O\!\left(\frac{1}{\mathbb{E}[\Theta_0]}\right).
\end{equation}

Now, we can include the remaining copies of the clock that are also ticking. Let $\tau_1(\vartheta_0), \ldots,\tau_N(\vartheta_0)$ be the tick-counting functions for the $N$ copies. Define the statistics of the observed number of ticks as:

\begin{equation}
\label{eq:app_A6}
    \widehat{\mu}_N := \frac{1}{N}\sum_{k=1}^N \tau_k(\vartheta_0),\ \ \ \ \widehat{\sigma}_N^2 := \frac{1}{N}\sum_{k=1}^N \big(\tau_k(\vartheta_0)-\widehat{\mu}_N\big)^2 .
\end{equation}

Because of the i.i.d. property of $\tau_k$, one can apply the strong law of large numbers,
\begin{equation}
\label{eq:app_A7}
\widehat{\mu}_N \xrightarrow[N\to\infty]{\text{a.s.}} \mathbb{E}[\beta(\vartheta_0)],
\ \ \ \ 
\widehat{\sigma}_N^2 \xrightarrow[N\to\infty]{\text{a.s.}} \operatorname{Var}[\beta(\vartheta_0)].
\end{equation}

Therefore, for fixed finite $\vartheta_0$,
\begin{equation}
\label{eq:app_A8}
\frac{\widehat{\mu}_N}{\widehat{\sigma}_N^2}
\xrightarrow[N\to\infty]{\text{a.s.}}
\frac{\mathbb{E}[\beta(\vartheta_0)]}{\operatorname{Var}[\beta(\vartheta_0)]}.
\end{equation}

To obtain an explicit rate, assume additionally that
\(\mathbb{E}[\beta(\vartheta_0)^4]<\infty\).
Then
\(\widehat{\mu}_N-\mathbb{E}[\beta(\vartheta_0)] = O_\mathbb{P}(N^{-1/2})\)
and
\(\widehat{\sigma}_N^2-\operatorname{Var}(\beta(\vartheta_0)) = O_\mathbb{P}(N^{-1/2})\).
Applying the delta method to $g(x,y)=x/y$ yields
\begin{equation}
\label{eq:app_A8b}
\frac{\widehat{\mu}_N}{\widehat{\sigma}_N^2}
=\frac{\mathbb{E}[\beta(\vartheta_0)]}{\operatorname{Var}(\beta(\vartheta_0))}
+ O_\mathbb{P}\!\left(\frac{1}{\sqrt{N}}\right).
\end{equation}

A combination of \cref{eq:app_A5b} and \cref{eq:app_A8b} concludes the proof.

\end{proof}

\subsection{Proofs of \Cref{th:2,th:3,th:4,th:5}}
\label{sec:app_subB}

\Cref{th:2} follows directly from \Cref{th:3,th:4}. Thus, we first prove these two results and then use them to prove \Cref{th:2}.

Let us start by restating and proving \Cref{th:4}.

\theoremfour*

\begin{proof}

Define the centred tick-time fluctuation $T_i':=T_i-\mu$ and the alternating sum
\begin{equation}
    S_n:= \sum_{i=1}^n (-1)^{i+1}T_i'.
\end{equation}
The increments $(-1)^{i+1}T_i'$ are independent and centred, and have variance $\sigma^2$. Define the stopping time

\begin{align}
\label{eq:app_B1}
    \tau := \min\bigl\{\,n:\,&(\Delta-\mu \ge S_n)\ \lor\ ( S_n > \Delta)\bigr\}.
\end{align}

By the stopping rule, on $\{\tau \text{ even}\}$ the process crosses the upper boundary $\Delta$, and on $\{\tau \text{ odd}\}$ it crosses the lower boundary $\Delta-\mu$.

Wald's identities give
\begin{align}
    \mathbb{E}[S_\tau] &= 0, \label{eq:wald1} \\
    \mathbb{E}[S_\tau^2] &= \sigma^2\mathbb{E}[\tau]. \label{eq:wald2}
\end{align}

Writing $P_{even} := P[\tau\text{ even}]$ and $P_{odd}:=P[\tau\text{ odd}]=1-P_{even}$, the law of total expectation gives

\begin{align}
\label{eq:app_B2}
    \sigma^2 \mathbb{E}[\tau]
    &= \mathbb{E}[S_\tau^2] \\
    \nonumber
    &= P_{even}\,\mathbb{E}[S_\tau^2\mid \tau\text{ even}]
     + P_{odd}\,\mathbb{E}[S_\tau^2\mid \tau\text{ odd}],
\end{align}

where we have used the standard notation $P[A|B]$ to denote the probability of $A$ conditioned on $B$. Define the overshoots
\begin{align}
    O_{+} &:= S_\tau-\Delta \ge 0 \quad \text{on } \{\tau\text{ even}\},\\
    O_{-} &:= (\Delta-\mu)-S_\tau \ge 0 \quad \text{on } \{\tau\text{ odd}\}.
\end{align}

Then, by the definition of the stopping problem \cref{eq:app_B1},

\begin{align*}
    \mathbb{E}[S_\tau^2\mid \tau\text{ even}]
    &= \Delta^2 + 2\Delta\,\mathbb{E}[O_{+}\mid \tau\text{ even}] + \mathbb{E}[(O_{+})^2\mid \tau\text{ even}],\\
    \mathbb{E}[S_\tau^2\mid \tau\text{ odd}]
    &=(\Delta-\mu)^2
    +2(\mu-\Delta)\,\mathbb{E}[O_{-}\mid \tau\text{ odd}]\\
    &\quad+\mathbb{E}[O_-^2\mid \tau\text{ odd}].
\end{align*}

Substituting into \cref{eq:app_B2} gives
\begin{equation}
\label{eq:wald_decomp}
    \sigma^2\mathbb{E}[\tau]
    = P_{even}\,\Delta^2 + P_{odd}\,(\Delta-\mu)^2 + \varepsilon,
\end{equation}

where the error term $\varepsilon$ is given by

\begin{equation}
\label{eq:def_eps}
\begin{aligned}
    \varepsilon
    &:= 2\,(P_{even}\Delta\,\mathbb{E}[O_{+}\mid \tau\text{ even}]
    + P_{odd}(\mu-\Delta)\,\mathbb{E}[O_{-}\mid \tau\text{ odd}])\\
    &\quad+ P_{even}\,\mathbb{E}[O_+^2\mid \tau\text{ even}]\\
    &\quad+ P_{odd}\,\mathbb{E}[O_-^2\mid \tau\text{ odd}].
\end{aligned}
\end{equation}

Since $0\le \Delta\le \mu$, we have $\varepsilon\ge 0$, so \cref{eq:wald_decomp} immediately implies

\begin{equation}
\label{eq:tau_lower}
    \mathbb{E}[\tau]
    \ge \frac{P_{even}\,\Delta^2 + P_{odd}\,(\Delta-\mu)^2}{\sigma^2}.
\end{equation}

Since $S_{\tau-1}$ is still inside the interval $(\Delta-\mu,\Delta]$, the overshoot at time $\tau$ is no larger than the absolute size of the final centred increment:
\begin{equation}
    O_+ \le |T_\tau'| \quad \text{on } \{\tau\text{ even}\},\qquad
    O_- \le |T_\tau'| \quad \text{on } \{\tau\text{ odd}\}.
\end{equation}
Therefore,

\begin{align*}
    A_1&:=P_{even}\mathbb{E}[O_+\mid \tau\text{ even}]
       +P_{odd}\mathbb{E}[O_-\mid \tau\text{ odd}]
       \le \mathbb{E}[|T_\tau'|],\\
    A_2&:=P_{even}\mathbb{E}[O_+^2\mid \tau\text{ even}]
       +P_{odd}\mathbb{E}[O_-^2\mid \tau\text{ odd}]
       \le \mathbb{E}[|T_\tau'|^2].
\end{align*}

Moreover,

\begin{equation}
    |T_\tau'|\le \max_{1\le i\le \tau}|T_i'|
    \le \left(\sum_{i=1}^{\tau}|T_i'|^p\right)^{1/p}.
\end{equation}
Using Jensen's inequality and the fact that $\{\tau\ge i\}$ is determined before observing $T_i'$ and hence is independent of $T_i'$, and applying the assumption from the main text regarding the centred absolute moments convergence assumption (\cref{eq:moment_assumption}), 

\begin{equation}
    \mathbb{E}\!\left[\left|\frac{T_i'}{\sigma}\right|^p\right]\le C_p,
\end{equation}
we obtain
\begin{equation}
\label{eq:centered_random_sum}
    \mathbb{E}\!\left[\sum_{i=1}^{\tau}|T_i'|^p\right]
    =\sum_{i=1}^{\infty}\mathbb{E}\!\left[|T_i'|^p\mathbb{I}_{\{\tau\ge i\}}\right]
    \le C_p\sigma^p\mathbb{E}[\tau].
\end{equation}
Consequently,
\begin{equation}
\label{eq:over_final}
    A_1\le C_p^{1/p}\sigma\,\mathbb{E}[\tau]^{1/p},
    \qquad
    A_2\le C_p^{2/p}\sigma^2\,\mathbb{E}[\tau]^{2/p}.
\end{equation}

These estimates also imply that $\mathbb{E}[\tau]/R$ remains bounded uniformly for large $R$. Define $x:=\mathbb{E}[\tau]/R$. We want to show that $x$ is bounded above for any sufficiently large $R$. Indeed, using $\Delta\in[0,\mu]$, \cref{eq:wald_decomp} and \cref{eq:over_final} yield
\begin{equation}
    x\le 1+2C_p^{1/p}R^{-1/2+1/p}x^{1/p}
      +C_p^{2/p}R^{-1+2/p}x^{2/p}.
\end{equation}
Therefore, since $p>2$, the right-hand side grows sublinearly in $x$, and hence there must exist a constant $K_p>0$ independent of $R$ and $\delta$ (the inequality above holds uniformly for $\Delta\in[0,\mu]$) such that $x\le K_p$. Substituting this bound back into \cref{eq:over_final} gives
\begin{align}
\label{eq:overshootbounds}
    A_1 &\le \kappa_{1,p}\frac{\mu}{R^{1/2-1/p}},\\
    A_2 &\le \kappa_{2,p}\frac{\mu^2}{R^{1-2/p}}.
\end{align}

The coefficients $\kappa_{1,p}, \kappa_{2,p}$ arise from $C_p$ in \cref{eq:centered_random_sum} and from $K_p$ that has just been defined. To conclude the upper bound, substitute \cref{eq:overshootbounds} into \cref{eq:def_eps} as follows:

\begin{align}
\label{eq:bound_eps}
    \varepsilon
    &\le 2\mu A_1+A_2 \nonumber\\
    &\le \frac{2\kappa_{1,p}\mu^2}{R^{1/2-1/p}}
      +\frac{\kappa_{2,p}\mu^2}{R^{1-2/p}}.
\end{align}

From \cref{eq:wald_decomp}, we therefore obtain

\begin{align}
    &0\le \sigma^2\mathbb{E}[\tau] - (P_{even}\,\Delta^2 + P_{odd}(\Delta-\mu)^2) \\
    &\nonumber\le \frac{2\kappa_{1,p}\mu^2}{R^{1/2-1/p}}
      +\frac{\kappa_{2,p}\mu^2}{R^{1-2/p}}.
\end{align}

Dividing both sides by $\mu^2$ yields

\begin{align}
\label{eq:wald_upper_th4}
    &0 \le \frac{\mathbb{E}[\tau]}{R} - \left(P_{even}\,\delta^2 + P_{odd}(\delta-1)^2\right) \\
    &\nonumber \le \frac{2\kappa_{1,p}}{R^{1/2-1/p}} + \frac{\kappa_{2,p}}{R^{1-2/p}}.
\end{align}

\Cref{eq:wald_upper_th4} provides a bound on the relationship between $\mathbb{E}[\tau]$ and $R$ in terms of $P_{odd}$, which can be experimentally measured to arbitrary precision, in the case of a known $\delta$. Nevertheless, one could argue that in order to characterise $\delta$ perfectly one would need an ideal reference frame. However, a bound for $\delta$ that only depends on operationally measurable quantities can be obtained as follows. 

By the Optional Stopping Theorem applied to \cref{eq:app_B1},
we have
\begin{equation}
\label{eq:ost0}
    \mathbb{E}[S_\tau] = \mathbb{E}[S_0] = 0.
\end{equation}

Therefore, as shown in \cref{eq:wald1},

\begin{align}
\label{eq:app_B8}
    0=\mathbb{E}[S_\tau]
    &= P_{even}\,\mathbb{E}[S_\tau \mid \tau \text{ even}]\\
    &\quad + P_{odd}\,\mathbb{E}[S_\tau \mid \tau \text{ odd}],
\end{align}

which can again be expressed in terms of the overshoots as

\begin{align*}
0 &= P_{even} (\Delta+\mathbb{E}\left[O_+\mid \tau\text{ even}\right])\\
  &\quad+ P_{odd} (\Delta - \mu-\mathbb{E}\left[O_-\mid \tau\text{ odd}\right]).
\end{align*}

Equivalently,
\begin{equation}
    \delta-P_{odd}
    =\frac{P_{odd}\mathbb{E}[O_-\mid \tau\text{ odd}]
    -P_{even}\mathbb{E}[O_+\mid \tau\text{ even}]}{\mu}.
\end{equation}

Using \cref{eq:overshootbounds}, we obtain

\begin{equation}
\label{eq:last_th3}
    \left|\delta - P_{odd}\right| \le \frac{\kappa_p}{R^{1/2-1/p}}.
\end{equation}

Now that we have this relationship between $\delta$ and $P_{odd}$, we can use it to refine \cref{eq:wald_upper_th4}. First, note that
\begin{equation}
    P_{even}\delta^2+P_{odd}(\delta-1)^2
    =P_{even}P_{odd}+(\delta-P_{odd})^2.
\end{equation}
Combining this identity with \cref{eq:wald_upper_th4,eq:last_th3} gives
\begin{equation}
    0 \le \frac{\mathbb{E}[\tau]}{R}-P_{even}P_{odd}
    \le \frac{K_{1,p}}{R^{1/2-1/p}}+\frac{K_{2,p}}{R^{1-2/p}},
\end{equation}

Relabelling the coefficients recovers the main result from \Cref{th:4}.

\end{proof}

For completeness, recall the statement of \Cref{th:3}:

\theoremthree*

\begin{proof}
This corresponds exactly to \cref{eq:last_th3}, so the result is already included in the proof of \Cref{th:4}. Note that \Cref{th:3} also implies that 

\begin{equation}
    \left|\delta - P_{odd}\right| = O \left(R^{-1/2+1/p}\right).
\end{equation}

\end{proof}

We can now use it to prove \Cref{th:2}:

\theoremtwo*

\begin{proof}
By letting $R\to\infty$ in \Cref{th:3}, we get:

\begin{equation}
\label{eq:app_B13}
    \delta = P_{odd} + O\left(\frac{1}{R^{1/2-1/p}}\right).
\end{equation}

Since $P_{odd} + P_{even} = 1$, \cref{eq:app_B13} implies
$P_{odd}P_{even}=\delta(1-\delta)+O(R^{-(1/2-1/p)})$. Substituting this into \Cref{th:4} concludes the proof. \Cref{th:2} also implies that 

\begin{equation}
    \frac{\mathbb{E}[\tau]}{R} = \delta(1-\delta) + O \left(R^{-1/2+1/p}\right).
\end{equation}
\end{proof}

Lastly, let us prove \Cref{th:5}. This result generalises the first-moment statements above to arbitrary moments of $\tau$. Its proof is self-contained (in particular, it relies neither on Wald's identities nor on \Cref{th:4}) and is a discrete, quantitative version of the following classical martingale argument for Brownian motion. If $(W_t)_{t\ge0}$ is a standard Brownian motion started at $x\in[0,1]$ and $T$ denotes its first exit time from $[0,1]$, then the moments $u_j(x):=\mathbb{E}_x[T^j]$ solve exactly the Dirichlet hierarchy defining the polynomials $P_j$ in \Cref{th:5}, so that $u_j=P_j$. The classical proof applies the optional stopping theorem to the martingale $h(t\wedge T,W_{t\wedge T})$, where $h(t,x):=\mathbb{E}_x[(t+T)^m]$ solves the backward heat equation $\partial_t h+\tfrac12\partial_x^2h=0$ with boundary data $h(t,0)=h(t,1)=t^m$ and initial slice $h(0,x)=P_m(x)$, yielding $\mathbb{E}_x[T^m]=h(0,x)=P_m(x)$. Our ticking walk, once suitably normalised, has centred increments of variance exactly $1/R$; on time scales of order $R$ it mimics such a Brownian motion, $\tau/R$ mimics $T$, and $h(n/R,\cdot)$ evaluated along the walk is an approximate martingale with explicitly controlled drift. The proof below consists of the corresponding quantitative statements: an explicit construction of $h$ (\Cref{lem:heat_polynomial}), a one-step drift bound (\Cref{lem:drift}), an a priori integrability statement for $\tau$ that justifies all limiting operations (\Cref{lem:tau_moments}), and overshoot estimates quantifying the failure of the walk to stop exactly on the boundary (\Cref{lem:overshoot}).

Before we start the proof, we need to introduce some definitions. In the following, $K$ denotes a finite positive constant depending only on $m$, $p$ and $C_p$ (in particular, independent of $R$, $\delta$ and $n$), whose value may change from line to line.

We first normalise the walk so that it starts at $\delta$ and is stopped upon leaving the unit interval. Recall that $S_n=\sum_{i=1}^n(-1)^{i+1}T_i'$ and $\tau=\min\{n\ge1: \Delta < S_n\le\Delta-\mu\}$, and define
\begin{equation}
\label{eq:X_definition}
    X_n:=\delta-\frac{S_n}{\mu},
    \qquad
    \mathcal{F}_n:=\sigma(T_1',\ldots,T_n').
\end{equation}
Then
\begin{align}
\label{eq:tau_X}
X_0&=\delta\in(0,1),\nonumber\\
\tau&=\min\{n\ge1:\ X_n \notin (0,1] \},
\\X_n&\in[0,1)\text{ for } \{\tau>n\}.\nonumber
\end{align}
Since $T_{n+1}'$ is independent of $\mathcal{F}_n$, the increments $\eta_{n+1}:=X_{n+1}-X_n=(-1)^{n+1}T_{n+1}'/\mu$ satisfy
\begin{align}
\label{eq:increment_moments}
    \mathbb{E}[\eta_{n+1}\mid\mathcal{F}_n]&=0,
    \nonumber\\
    \mathbb{E}[\eta_{n+1}^2\mid\mathcal{F}_n]&=\frac{1}{R},
    \\
    \mathbb{E}[|\eta_{n+1}|^q\mid\mathcal{F}_n]&\le\frac{C_p^{q/p}}{R^{q/2}},
    \qquad 1\le q\le 2m.
    \nonumber
\end{align}
The first two identities are immediate from $\mathbb{E}[T_{n+1}']=0$ and $\mathbb{E}[T_{n+1}'^2]=\sigma^2=\mu^2/R$. The last bound is obtained as follows: since $q\le2m<p$ by the hypothesis of the theorem, the map $u\mapsto u^{q/p}$ is concave on $[0,\infty)$, so Jensen's inequality applies in the form $\mathbb{E}[Z^{q/p}]\le(\mathbb{E}[Z])^{q/p}$ for $Z\ge0$, giving

\begin{align}
    \mathbb{E}\!\left[\left|\frac{T_{n+1}'}{\sigma}\right|^{q}\right]
    & =\mathbb{E}\!\left[\left(\left|\frac{T_{n+1}'}{\sigma}\right|^{p}\right)^{q/p}\right] \le\left(\mathbb{E}\!\left[\left|\frac{T_{n+1}'}{\sigma}\right|^{p}\right]\right)^{q/p} \nonumber
    \\ & \le C_p^{q/p} \nonumber,
\end{align}

where the last step is \cref{eq:moment_assumption}. Multiplying by $(\sigma/\mu)^q=R^{-q/2}$ yields the third line of \cref{eq:increment_moments}.

\begin{lemma}[The caloric polynomial]
\label{lem:heat_polynomial}
Let $P_0,\ldots,P_m$ be as in \Cref{th:5} and define
\begin{equation}
\label{eq:h_definition}
    h(t,x):=\sum_{j=0}^{m}\binom{m}{j}\,t^{m-j}P_j(x),
\end{equation}
a real polynomial in $(t,x)$. Then, for all $t\ge 0$,

\begin{enumerate}
    \item[(i)] each $P_j$ exists, is unique, and is a polynomial of degree exactly $2j$, which, for $x \in \mathbb{R}$, is given recursively by
    \begin{align}
    \label{eq:P_recursion}
        P_j(x)={}&2j(1-x)\int_0^x u\,P_{j-1}(u)\,du
        \nonumber\\
        &+2jx\int_x^1(1-u)\,P_{j-1}(u)\,du;
    \end{align}
    \item[(ii)] $\partial_t h+\tfrac12\partial_x^2h=0$ and
    \begin{equation}
        h(t,0)=h(t,1)=t^m,
        \qquad
        h(0,x)=P_m(x)
    \end{equation}
    for all $x \in \mathbb{R}$;
    \item[(iii)] for all integers $k,r\ge0$, the function $\partial_t^k\partial_x^r h(t,x)$ is a polynomial in $t$ of degree at most $m-k-\lceil r/2\rceil$ (vanishing identically if this quantity is negative), whose coefficients are bounded by $K$ uniformly for $x\in[0,1]$;
    \item[(iv)] for every $1\le j\le m$, $v\ge0$ and $x\in[-v,0]\cup[1,1+v]$,
    \begin{equation}
        |P_j(x)|\le K\,\bigl(v+v^{2j}\bigr).
    \end{equation}
\end{enumerate}
\end{lemma}

\begin{proof}
(i) We induct on $j$, starting from $P_0\equiv1$. For $j\ge1$, the right-hand side of \cref{eq:P_recursion} vanishes at $x=0$ and $x=1$, its first derivative equals
\begin{equation}
    2j\left(-\int_0^x u\,P_{j-1}(u)\,du+\int_x^1(1-u)\,P_{j-1}(u)\,du\right),
\end{equation}
and hence its second derivative equals $-2jP_{j-1}(x)$; it therefore solves the boundary value problem of \Cref{th:5}. It is the unique solution: the difference of two solutions has vanishing second derivative, hence is affine, and vanishes at $0$ and $1$, hence is identically zero. Finally, if $P_{j-1}$ has degree $2(j-1)$ with leading coefficient $c_{j-1}\neq0$, then $P_j''=-2jP_{j-1}$ forces $P_j$ to have degree $2j$ with leading coefficient $c_j=-c_{j-1}/(2j-1)\neq0$.

(ii) Since $P_j(0)=P_j(1)=0$ for $j\ge1$ and $P_0\equiv1$, only the $j=0$ term of \cref{eq:h_definition} survives at $x\in\{0,1\}$, giving $h(t,0)=h(t,1)=t^m$; at $t=0$ only the $j=m$ term survives, giving $h(0,x)=P_m(x)$. Moreover,
\begin{align}
    \partial_t h(t,x) \label{eq:cancel1}
    &=\sum_{j=0}^{m-1}\binom{m}{j}(m-j)\,t^{m-j-1}P_j(x),
    \\
    \frac12\partial_x^2h(t,x) \label{eq:cancel2}
    &=-\sum_{j=1}^{m}\binom{m}{j}\,j\,t^{m-j}P_{j-1}(x)
    \\
    \nonumber&=-\sum_{j=0}^{m-1}\binom{m}{j+1}(j+1)\,t^{m-j-1}P_j(x),
\end{align}
where we used $\tfrac12P_j''=-jP_{j-1}$ and then relabelled $j-1$ as $j$. The identity $\binom{m}{j}(m-j)=\binom{m}{j+1}(j+1)$ shows that \cref{eq:cancel1,eq:cancel2} cancel term by term.

(iii) Let us denote the $r$-th derivative of $P_j$ by $P_j^{(r)}$. By (i), $P_j^{(r)}\equiv0$ whenever $r>2j$, so in $\partial_t^k\partial_x^rh$ only the indices $\lceil r/2\rceil\le j\le m-k$ contribute, with powers $t^{m-j-k}$ of degree at most $m-k-\lceil r/2\rceil$. The coefficients involve only binomial factors and the derivatives $P_j^{(r)}$, which are fixed polynomials once $m$ is fixed, evaluated on the compact set $[0,1]$; they are therefore bounded by $K$.

(iv) Let $x=1+u$ with $0\le u\le v$. Since $P_j(1)=0$ and $\deg P_j=2j$, Taylor's formula at the root $1$ gives $P_j(1+u)=\sum_{\ell=1}^{2j}P_j^{(\ell)}(1)\,u^\ell/\ell!$, whence
\begin{equation}
    |P_j(1+u)|\le K\sum_{\ell=1}^{2j}v^\ell\le K\,\bigl(v+v^{2j}\bigr),
\end{equation}
because $v^\ell\le v+v^{2j}$ for every $1\le\ell\le2j$ (consider $v\le1$ and $v>1$ separately). The case $x\in[-v,0]$ is identical, using the root at $0$.
\end{proof}

\begin{lemma}[Drift bound]
\label{lem:drift}
Define
\begin{equation}
\label{eq:M_definition}
    M_n:=h\!\left(\frac{n}{R},X_n\right),
    \qquad n\ge0,
\end{equation}
so that $M_0=h(0,\delta)=P_m(\delta)$. Then, for every $n\ge0$ and every $R\ge1$, on the event $\{\tau>n\}$,
\begin{equation}
\label{eq:drift_bound}
    \bigl|\mathbb{E}[M_{n+1}-M_n\mid\mathcal{F}_n]\bigr|
    \le\frac{K}{R^{3/2}}
    \left(1+\left(\frac{n}{R}\right)^{m-1}\right).
\end{equation}
\end{lemma}

\begin{proof}
Since $h$ is a polynomial, Taylor's formula in the two variables is a finite, exact identity:
\begin{equation}
    M_{n+1}
    =\sum_{k,r\ge0}\frac{R^{-k}\,\eta_{n+1}^{\,r}}{k!\,r!}\,
    \partial_t^k\partial_x^r h\!\left(\frac{n}{R},X_n\right),
\end{equation}
with nonzero terms only for $k\le m$ and $r\le2m$, by \Cref{lem:heat_polynomial}(iii). The derivatives are $\mathcal{F}_n$-measurable, and on $\{\tau>n\}$ each term is integrable, because there $X_n\in[0,1)$ while $\eta_{n+1}$ has conditional moments up to order $2m$ by \cref{eq:increment_moments}. We take $\mathbb{E}[\,\cdot\mid\mathcal{F}_n]$ term by term. The $(k,r)=(0,0)$ term is $M_n$. Every term with $r=1$ vanishes because the increments are centred. The two terms $(k,r)=(1,0)$ and $(0,2)$ jointly contribute
\begin{equation}
    \frac{1}{R}\left(\partial_t h+\frac12\,\partial_x^2h\right)
    \!\left(\frac{n}{R},X_n\right)=0
\end{equation}
by \Cref{lem:heat_polynomial}(ii), where we used that $\mathbb{E}[\eta_{n+1}^2\mid\mathcal{F}_n]=1/R$ holds exactly. Every remaining term has $k\ge2,r=0$, or $k\ge1,r\ge2$, or $k=0,r\ge3$; in all three cases $k+r/2\ge3/2$ and $m-k-\lceil r/2\rceil\le m-1$. Hence, by \Cref{lem:heat_polynomial}(iii) (applicable since $X_n\in[0,1)$) and \cref{eq:increment_moments}, the conditional expectation of each such term is bounded in absolute value by
\begin{equation}
    K\,R^{-k-r/2}\left(1+\left(\frac{n}{R}\right)^{m-1}\right)
    \le \frac{K}{R^{3/2}}\left(1+\left(\frac{n}{R}\right)^{m-1}\right)
\end{equation}
for $R\ge1$, where we used $t^a\le1+t^{m-1}$ for all $t\ge0$ and $0\le a\le m-1$. Summing over the boundedly many pairs $(k,r)$ proves \cref{eq:drift_bound}.
\end{proof}

To provide some context to the reader, \Cref{lem:drift}, applied to our problem, allows us to understand $M_n$ as a stochastic process that behaves almost like a martingale: while in the case of a martingale $X_n$ one would have, by definition, $\mathbb{E}[X_{n+1}|\mathcal{F}_n] - \mathbb{E}[X_{n}] = 0$, in our case this quantity is bounded above by a term that becomes arbitrarily small for sufficiently precise clocks (sufficiently large $R$).

\begin{lemma}[A priori integrability]
\label{lem:tau_moments}
$\tau<\infty$ almost surely, and $\mathbb{E}[\tau^s]<\infty$ for every $s>0$.
\end{lemma}

This statement is purely qualitative: the constants implicit in its proof may depend on the ticking distribution, and hence on $R$. This is harmless, since the lemma is used only to justify limiting operations, while all quantitative bounds below are uniform.

\begin{proof}
The two-step displacements
\begin{equation}
    D_i:=X_{2i}-X_{2i-2}=\frac{T_{2i}'-T_{2i-1}'}{\mu},
    \qquad i\ge1,
\end{equation}
are independent and identically distributed, centred, and non-degenerate, since $\operatorname{Var}(D_1)=2/R>0$. Consequently $\mathbb{P}[D_1>0]>0$ (otherwise $D_1\le0$ almost surely together with $\mathbb{E}[D_1]=0$ would force $D_1=0$ almost surely), and by continuity from below there exists $a>0$ with $q:=\mathbb{P}[D_1>a]>0$. Let $k:=\lceil1/a\rceil$, and fix $i\ge0$. On the event $\{\tau>2ik\}$ one has $X_{2ik}\in[0,1)$; if moreover $D_l>a$ for every $l\in\{ik+1,\ldots,(i+1)k\}$ (an event of probability $q^k$, independent of $\mathcal{F}_{2ik}$), then
\begin{equation}
    X_{2(i+1)k}=X_{2ik}+\sum_{l=ik+1}^{(i+1)k}D_l> ka\ge1,
\end{equation}
so that $\tau\le2(i+1)k$. Since $\{\tau>2ik\}\in\mathcal{F}_{2ik}$, it follows that
\begin{equation}
    \mathbb{P}[\tau>2(i+1)k]\le\bigl(1-q^k\bigr)\,\mathbb{P}[\tau>2ik],
\end{equation}
and inductively $\mathbb{P}[\tau>2ik]\le(1-q^k)^{i}$ for all $i\ge0$. The rate of decay of the tail of $\tau$ implies both assertions.
\end{proof}

\begin{lemma}[Overshoot bounds]
\label{lem:overshoot}
Set $V:=|T_\tau'|/\mu$, which is well defined by \Cref{lem:tau_moments}. Then
\begin{equation}
\label{eq:exit_position}
    X_\tau\in[-V,0)\cup[1,1+V],
\end{equation}
\nonumber
and for every real $1\le q\le 2m$,
\begin{equation}
\nonumber
\label{eq:overshoot_moments}
    \mathbb{E}[V^q]\le\frac{K}{R^{q/2}}\,\mathbb{E}[\tau]^{q/p}.
\end{equation}
\end{lemma}

\begin{proof}
Since $X_{\tau-1}\in[0,1)$ (for $\tau=1$ because $X_0=\delta$) and $|X_\tau-X_{\tau-1}|=V$, we have: if $X_\tau\ge1$, then
\begin{equation}
\nonumber
0\le X_\tau-1\le X_\tau-X_{\tau-1}\le V;
\end{equation}
 and if $X_\tau<0$, then 
 \begin{equation}
 \nonumber
 0<-X_\tau\le X_{\tau-1}-X_\tau\le V.
 \end{equation}
This proves \cref{eq:exit_position}. Next, $|T_\tau'|^p\le\sum_{i=1}^{\tau}|T_i'|^p$, so that

\begin{equation}
    V^q=\left(\frac{|T_\tau'|^p}{\mu^p}\right)^{q/p}
    \le\left(\frac{1}{\mu^p}\sum_{i=1}^{\tau}|T_i'|^p\right)^{q/p}.
\end{equation}
Write $Z:=\mu^{-p}\sum_{i=1}^{\tau}|T_i'|^p\ge0$, so that the previous display reads $V^q\le Z^{q/p}$. Since $q\le2m<p$, the map $u\mapsto u^{q/p}$ is concave on $[0,\infty)$, so Jensen's inequality applies in the form $\mathbb{E}[Z^{q/p}]\le(\mathbb{E}[Z])^{q/p}$. Together with Tonelli's theorem, that allows us to interchange the sum with the expected value giving rise to $\mathbb{E}[Z]=\mu^{-p}\sum_{i=1}^{\infty}\mathbb{E}[|T_i'|^p\,\mathbb{I}_{\{\tau\ge i\}}]$, this yields

\begin{equation}
    \mathbb{E}[V^q]
    \le\left(\frac{1}{\mu^p}\sum_{i=1}^{\infty}
    \mathbb{E}\bigl[|T_i'|^p\,\mathbb{I}_{\{\tau\ge i\}}\bigr]\right)^{q/p}.
\end{equation}
Since $\{\tau\ge i\}=\{\tau>i-1\}\in\mathcal{F}_{i-1}$ is independent of $T_i'$, \cref{eq:moment_assumption} bounds each summand by $C_p(\sigma/\mu)^p\,\mathbb{P}[\tau\ge i]$, and $\sum_{i\ge1}\mathbb{P}[\tau\ge i]=\mathbb{E}[\tau]$. As $\sigma/\mu=R^{-1/2}$, \cref{eq:overshoot_moments} follows.
\end{proof}

Let us restate \Cref{th:5} before proving it.

\theoremfive*

\begin{proof}[Proof of \Cref{th:5}]
Assume $R\ge1$ throughout, and write
\begin{equation}
    y:=\mathbb{E}\!\left[\left(\frac{\tau}{R}\right)^m\right],
\end{equation}
which is finite by \Cref{lem:tau_moments}. Let us first outline the argument. The process $M_n$ of \Cref{lem:drift} connects the two quantities that the theorem compares. On the one hand, $M_0=P_m(\delta)$ and $M_n$ is almost a martingale by \Cref{lem:drift}, so $\mathbb{E}[M_\tau]$ should be close to $P_m(\delta)$. On the other hand, $h(t,x)=t^m$ whenever $x$ lies exactly on the boundary $\{0,1\}$, and $X_\tau$ misses the boundary only by an overshoot controlled by $V$, so $M_\tau$ should be close to $(\tau/R)^m$. We prove quantitative versions of these two approximations, namely \cref{eq:terminal_error,eq:optional_stopping} below. Both error bounds involve the unknown moment $y$ itself; the third and final estimate, \cref{eq:uniform_moment}, resolves this circularity by showing that $y$ is bounded uniformly in $R$ and $\delta$. The theorem then follows by combining the three estimates.

We begin with the comparison of $M_\tau$ and $(\tau/R)^m$. The goal of this part of the proof is the bound
\begin{equation}
\label{eq:terminal_error}
    \mathbb{E}\!\left[\left|M_\tau-\left(\frac{\tau}{R}\right)^m\right|\right]
    \le\frac{K}{R^{\frac12-\frac1p}}\,\bigl(1+y^{\theta}\bigr),
\end{equation}
where
\begin{equation}
\label{eq:theta_def}
    \theta:=1-\frac{1}{m}\left(1-\frac{2}{p}\right)<1;
\end{equation}
the strict inequality $\theta<1$ holds precisely because $p>2$, and will be essential at the end of the proof. To prove \cref{eq:terminal_error}, note first that the $j=0$ term of \cref{eq:h_definition} is $t^m$, so
\begin{equation}
    M_\tau-\left(\frac{\tau}{R}\right)^m
    =\sum_{j=1}^{m}\binom{m}{j}
    \left(\frac{\tau}{R}\right)^{m-j}P_j(X_\tau).
\end{equation}
By \cref{eq:exit_position}, the exit position $X_\tau$ lies within distance $V$ of the boundary points $\{0,1\}$, so \Cref{lem:heat_polynomial}(iv), applied with $v=V$, bounds each $|P_j(X_\tau)|$ and gives the pointwise estimate
\begin{equation}
\label{eq:terminal_comparison}
    \left|M_\tau-\left(\frac{\tau}{R}\right)^m\right|
    \le K\sum_{j=1}^{m}
    \left(\frac{\tau}{R}\right)^{m-j}
    \bigl(V+V^{2j}\bigr).
\end{equation}
We now take expectations in \cref{eq:terminal_comparison} term by term. Fix $1\le j\le m$ and $b\in\{1,2j\}$. For $j<m$, H\"older's inequality with the conjugate exponents $\frac{m}{m-j}$ and $\frac{m}{j}$, in the form
\begin{equation}
    \mathbb{E}[AB]\le
    \mathbb{E}\bigl[A^{\frac{m}{m-j}}\bigr]^{\frac{m-j}{m}}\,
    \mathbb{E}\bigl[B^{\frac{m}{j}}\bigr]^{\frac{j}{m}}
    \qquad (A,B\ge0),
\end{equation}
applied to $A=(\tau/R)^{m-j}$ and $B=V^{b}$ yields
\begin{equation}
\label{eq:holder_term}
    \mathbb{E}\!\left[\left(\frac{\tau}{R}\right)^{m-j}V^b\right]
    \le y^{\frac{m-j}{m}}\;
    \mathbb{E}\bigl[V^{bm/j}\bigr]^{\frac{j}{m}};
\end{equation}
for $j=m$ the two sides of \cref{eq:holder_term} coincide, so it holds trivially. The exponent $bm/j$ lies in $[1,2m]$ because $1\le b\le2j$, so the overshoot bound \cref{eq:overshoot_moments} applies to the right-hand side of \cref{eq:holder_term}. That bound still involves $\mathbb{E}[\tau]$, which we remove in favour of $y$: Jensen's inequality for the convex map $u\mapsto u^m$, in the form $\mathbb{E}[\tau]^m\le\mathbb{E}[\tau^m]$, gives
\begin{equation}
    \mathbb{E}[\tau]\le\mathbb{E}[\tau^m]^{1/m}=R\,y^{1/m}.
\end{equation}
Substituting this into \cref{eq:overshoot_moments} with $q=bm/j$ and raising the result to the power $j/m$, we obtain
\begin{equation}
    \mathbb{E}\bigl[V^{bm/j}\bigr]^{\frac{j}{m}}
    \le K\,R^{-b\left(\frac12-\frac1p\right)}\;y^{\frac{b}{mp}}.
\end{equation}
It remains to collect the powers of $R$ and $y$. Since $b\ge1$ and $R\ge1$, the factor of $R$ is at most $R^{-(1/2-1/p)}$. The total exponent of the $y$ terms after substituting this in \cref{eq:holder_term} is
\begin{equation}
    \frac{m-j}{m}+\frac{b}{mp}
    \le 1-\frac{j}{m}\left(1-\frac{2}{p}\right)
    \le\theta,
\end{equation}
where the first inequality uses $b\le2j$, the second uses $j\ge1$ and $\theta$ is defined in \cref{eq:theta_def}. Using $y^{e}\le1+y^{\theta}$ for $0\le e\le\theta$ and summing the finitely many terms of \cref{eq:terminal_comparison} proves \cref{eq:terminal_error}. Since $y<\infty$, a byproduct of \cref{eq:terminal_error} is that $M_\tau$ is integrable, a fact we will need shortly.

We next quantify how close $M_n$ is to being a martingale when run until time $\tau$. The goal of this part of the proof to prove the bound
\begin{equation}
\label{eq:optional_stopping}
    \bigl|\mathbb{E}[M_\tau]-P_m(\delta)\bigr|
    \le\frac{K}{\sqrt{R}}\,(1+y).
\end{equation}
This is an optional stopping statement; since $\tau$ is unbounded, we first establish the estimate for the bounded stopping times $\tau\wedge L$ and then remove the cutoff. Fix $L\in\mathbb{N}$. Since $\{\tau>n\}\in\mathcal{F}_n$ and $M_0=P_m(\delta)$,
\begin{equation}
    M_{\tau\wedge L}-P_m(\delta)
    =\sum_{n=0}^{L-1}\mathbb{I}_{\{\tau>n\}}\,(M_{n+1}-M_n),
\end{equation}
where each summand is integrable, as noted in the proof of \Cref{lem:drift}. Taking expectations, conditioning the $n$th summand on $\mathcal{F}_n$, and applying the drift bound \cref{eq:drift_bound},
\begin{align}
\label{eq:optional_stopping_N}
    \bigl|\mathbb{E}[M_{\tau\wedge L}]-P_m(\delta)\bigr|
    &\le\frac{K}{R^{3/2}}\,
    \mathbb{E}\!\left[\sum_{n=0}^{\tau-1}
    \left(1+\left(\frac{n}{R}\right)^{m-1}\right)\right]
    \nonumber\\
    &\le\frac{K}{\sqrt{R}}\,(1+y),
\end{align}
uniformly in $L$, where the last step uses
\begin{equation}
\nonumber
    \sum_{n=0}^{\tau-1}\left(1+\left(\frac{n}{R}\right)^{m-1}\right)
    \le \tau+\tau\left(\frac{\tau}{R}\right)^{m-1}
    \le 2R\left(1+\left(\frac{\tau}{R}\right)^{m}\right).
\end{equation}

To deduce \cref{eq:optional_stopping} from \cref{eq:optional_stopping_N}, we must find a way to replace $\mathbb{E}[M_{\tau\wedge L}]$ by $\mathbb{E}[M_\tau]$. For this, one has to take the limit $L\to\infty$ on the left hand side of~\cref{eq:optional_stopping_N}, followed by using the dominated convergence theorem to exchange $\lim_{L\to\infty}\mathbb{E}[M_{\tau\wedge L}]$ for $\mathbb{E}[\lim_{L\to\infty}M_{\tau\wedge L}]$. We now prove the necessary prerequisites for the dominated convergence theorem to hold in our case. To start with, pointwise convergence holds because $\tau<\infty$ almost surely (\Cref{lem:tau_moments}): for almost every outcome, $\tau\wedge L=\tau$, and hence $M_{\tau\wedge L}=M_\tau$, as soon as $L\ge\tau$. Thus, all that remains to prove is the existence of a dominating integrable function for the set of functions $\{|M_{\tau\wedge L}|\}_L$.

To show a dominating function, we distinguish the two possible values of $\tau\wedge L$. On $\{\tau\le L\}$ one has $\tau\wedge L=\tau$, so that $M_{\tau\wedge L}=M_\tau$ and therefore $|M_{\tau\wedge L}|\le|M_\tau|$. On $\{\tau>L\}$ one has $\tau\wedge L=L$, so that $M_{\tau\wedge L}=M_L=h(L/R,X_L)$; moreover, on this event the walk has not exited by time $L$, so $X_L\in[0,1)$. We may therefore bound $h(L/R,X_L)$ through the behaviour of $h$ on $[0,\infty)\times[0,1]$: by \Cref{lem:heat_polynomial} (iii) with $k=r=0$, for each fixed $x\in[0,1]$ the map $t\mapsto h(t,x)$ is a polynomial of degree at most $m$ whose coefficients are bounded by $K$ uniformly in $x$, whence

\begin{equation}
    \sup_{x\in[0,1]}|h(t,x)|\le K\,(1+t^m),
    \qquad t\ge0 .
\end{equation}
Using this together with $L<\tau$ on the event $\{\tau>L\}$,
\begin{equation}
    |M_L|
    \le K\left(1+\left(\frac{L}{R}\right)^m\right)
    \le K\left(1+\left(\frac{\tau}{R}\right)^m\right).
\end{equation}
Combining the two cases, for every $L$,
\begin{equation}
    |M_{\tau\wedge L}|\le K\left(1+\left(\frac{\tau}{R}\right)^m\right)+|M_\tau| .
\end{equation}
This dominating function on the right hand side is integrable: the first summand because $y<\infty$, and the second because $M_\tau$ was shown to be integrable after \cref{eq:terminal_error}. Dominated convergence now gives $\mathbb{E}[M_{\tau\wedge L}]\to\mathbb{E}[M_\tau]$ as $L\to\infty$, and \cref{eq:optional_stopping} follows from \cref{eq:optional_stopping_N}.

The right-hand sides of \cref{eq:terminal_error,eq:optional_stopping} still involve the unknown moment $y$. The purpose of the third and final estimate is to show that these two bounds, taken together, force
\begin{equation}
\nonumber
\label{eq:uniform_moment}
    y\le K
    \qquad\text{for all } R\ge R_0,
\end{equation}
for some $R_0\ge1$ depending only on $m$, $p$ and $C_p$. Since $(\tau/R)^m\ge0$,
\begin{equation}
    y
    =\mathbb{E}\!\left[\left(\frac{\tau}{R}\right)^m\right]
    \le\mathbb{E}[M_\tau]
    +\mathbb{E}\!\left[\left|M_\tau-\left(\frac{\tau}{R}\right)^m\right|\right].
\end{equation}
We bound the first term by \cref{eq:optional_stopping} together with $|P_m(\delta)|\le K$ (a consequence of \Cref{lem:heat_polynomial}(iii) with $k=r=0$, evaluated at $t=0$), and the second term by \cref{eq:terminal_error}; using also $R^{-(1/2-1/p)}\le1$, we arrive at
\begin{equation}
    y\le K+\frac{K}{\sqrt{R}}\,y+K\,y^{\theta}.
\end{equation}
Choose $R_0\ge1$, depending only on $m$, $p$ and $C_p$, such that $K/\sqrt{R_0}\le\tfrac12$. For $R\ge R_0$ we may subtract $y/2$ from both sides (legitimate precisely because $y<\infty$, by \Cref{lem:tau_moments}) and obtain
\begin{equation}
    y\le2K\bigl(1+y^{\theta}\bigr).
\end{equation}
Because $\theta<1$, this inequality is self-improving and yields \cref{eq:uniform_moment}: either $y\le1$, or $y>1$, in which case $2K\le2Ky^{\theta}$, hence $y\le4Ky^{\theta}$ and $y^{1-\theta}\le4K$.

The theorem now follows by combining the three estimates. By applying the triangle inequality while adding and subtracting $M_\tau$,
\begin{equation}
    \left|\frac{\mathbb{E}[\tau^m]}{R^m}-P_m(\delta)\right|
    \le\mathbb{E}\!\left[\left|\left(\frac{\tau}{R}\right)^m-M_\tau\right|\right]
    +\bigl|\mathbb{E}[M_\tau]-P_m(\delta)\bigr| ,
\end{equation}
so substituting \cref{eq:uniform_moment} into \cref{eq:terminal_error,eq:optional_stopping} gives
\begin{align}
\label{eq:app_B12}
    \left|\frac{\mathbb{E}[\tau^m]}{R^m}-P_m(\delta)\right|
    &\le\frac{K}{R^{\frac12-\frac1p}}+\frac{K}{\sqrt{R}}
    \nonumber\\
    &\le\frac{K}{R^{\frac12-\frac1p}}
\end{align}
for all $R\ge R_0$. Renaming $K$ as $\kappa_{m,p}$, which depends only on $m$, $p$ and $C_p$ and in particular is independent of $R$ and $\delta$, completes the proof.

Note that this also implies the more general relation

\begin{equation}
    \frac{\mathbb{E}[\tau^m]}{R^m}-P_m(\delta) = O\left(R^{-\frac{1}{2}+\frac{1}{p}}\right)
\end{equation}

\end{proof}

\section{Proof of the Allan-variance relation}
\label{sec:app_allan}

To prove \cref{eq:allanR}, we first introduce the standard definition of the Allan variance.
Let $y(t)$ denote the fractional-frequency deviation of the clock (so the instantaneous tick rate is $1+y(t)$ in our time units) and define the average fractional frequency over the $n$th interval of length $\tau$ by

\begin{equation}
    \bar{y}_n := \frac{1}{\tau}\int_{n\tau}^{(n+1)\tau} y(t)\,dt.
\end{equation}

The Allan variance for an averaging time $\tau>0$ is then defined by the standard formula

\begin{equation}
    \label{eq:allan_standard}
    \sigma_A^2(\tau) := \frac{1}{2}\,\mathbb{E}\!\big[(\bar{y}_{n+1}-\bar{y}_n)^2\big],
\end{equation}

for any integer $n\ge 0$ (stationarity is usually assumed so the expression is independent of $n$).

Next, introduce the tick-counting function $\kappa(a,b)$, defined as the number of ticks produced by the clock between times $t=a$ and $t=b$.
Observe that, with unit nominal tick rate,

\begin{equation}
    \kappa(0,\tau(n+1))-\kappa(0,\tau n)
    = \int_{n\tau}^{(n+1)\tau}\big(1+y(t)\big)\,dt
    = \tau\big(1+\bar{y}_n\big).
\end{equation}

Hence the second difference appearing in the tick-counting formulation satisfies

\begin{equation}
    \kappa(0,\tau(n+2)) - 2\kappa(0,\tau(n+1)) + \kappa(0,\tau n)
    = \tau\big(\bar{y}_{n+1}-\bar{y}_n\big).
\end{equation}

Substituting this identity into the alternative definition

\begin{equation}
    \label{eq:allandef}
    \sigma_A^2(\tau) := \frac{\mathbb{E}\!\big[\big(\kappa(0,\tau(n+2)) - 2\kappa(0,\tau(n+1)) + \kappa(0,\tau n)\big)^2\big]}{2\tau^2}
\end{equation}

gives

\begin{equation}
    \sigma_A^2(\tau)
    = \frac{1}{2\tau^2}\,\mathbb{E}\!\big[\tau^2(\bar{y}_{n+1}-\bar{y}_n)^2\big]
    = \frac{1}{2}\,\mathbb{E}\!\big[(\bar{y}_{n+1}-\bar{y}_n)^2\big],
\end{equation}

which coincides with the standard definition \cref{eq:allan_standard}. Stationarity guarantees independence of $n$; if the nominal tick rate differs from unity, the factor $\tau$ should be replaced by the nominal tick rate times $\tau$.

Since we are working specifically with independent ticking clocks, $\kappa$ satisfies both additivity $\kappa(a,c)=\kappa(a,b)+\kappa(b,c)$
and time-translation invariance $\kappa(a,b)=\kappa(a+d,b+d)$. Using these properties, one may rewrite the numerator in \cref{eq:allandef} and expand the expectation to obtain

\begin{align}
    \label{eq:allan2}
    2\tau^2\sigma_A^2(\tau)
    &=\mathbb{E}\!\left[\big(\kappa(\tau n,\tau(n+1))-\kappa(\tau(n+1),\tau(n+2))\big)^2\right]\nonumber\\
    &=2\Big(\operatorname{Var}\!\big(\kappa(0,\tau)\big)-\mathrm{Cov}\!\big(\kappa(0,\tau),\kappa(\tau,2\tau)\big)\Big),
\end{align}

where in the second line we used stationarity to identify the two one-interval variances and then expanded the square and collected variance/covariance terms.

It is convenient to express the covariance in terms of variances. From

\begin{align}
&\operatorname{Var}\!\big(\kappa(0,2\tau)\big)
\\
&=\operatorname{Var}\!\big(\kappa(0,\tau)+\kappa(\tau,2\tau)\big)
\\
&=2\operatorname{Var}\!\big(\kappa(0,\tau)\big)+2\mathrm{Cov}\!\big(\kappa(0,\tau),\kappa(\tau,2\tau)\big)
\end{align}

we obtain

\begin{equation}
\mathrm{Cov}\!\big(\kappa(0,\tau),\kappa(\tau,2\tau)\big)
=\tfrac{1}{2}\Big(\operatorname{Var}\!\big(\kappa(0,2\tau)\big)-2\operatorname{Var}\!\big(\kappa(0,\tau)\big)\Big).
\end{equation}

Under the assumptions of independent ticking clocks (finite second moment of the ticking distribution, stationarity and a finite correlation time), $\operatorname{Var}\!\big(\kappa(0,t)\big)$ grows asymptotically linearly in $t$. Hence

\begin{equation}
\operatorname{Var}\!\big(\kappa(0,2\tau)\big)-2\operatorname{Var}\!\big(\kappa(0,\tau)\big)=o\big(\operatorname{Var}\!\big(\kappa(0,\tau)\big)\big),
\end{equation}

so that the covariance is negligible compared with the variance when $\tau$ is large compared with the process correlation time. Therefore, from \cref{eq:allan2} we have, for large $\tau$,

\begin{equation}
    \tau^2\sigma_A^2(\tau)=\operatorname{Var}\!\big(\kappa(0,\tau)\big)+o\big(\operatorname{Var}\!\big(\kappa(0,\tau)\big)\big).
\end{equation}

Finally, note that in our notation $\kappa(0,t)=\tau(t)$ (tick count up to time $t$). Therefore, by \cref{eq:app_A3},
$\operatorname{Var}\!\big(\kappa(0,\tau)\big)=\dfrac{\sigma^2}{\mu^3}\,\tau+O(1)$, and we obtain the large-$\tau$ behaviour

\begin{equation}
    \sigma_A^2(\tau)\sim\frac{\sigma^2}{\mu^3}\,\frac{1}{\tau}.
\end{equation}

This completes the derivation (under the stated hypotheses of stationarity, finite second moments and linear variance growth).

\qed

\nocite{*}
\AtBeginEnvironment{thebibliography}{\selectfont}
\bibliography{apssamp}% 

\end{document}